\documentclass[12pt]{article}
\usepackage[margin=1.3in]{geometry}
\usepackage{graphicx}
\usepackage[round, authoryear]{natbib}
\providecommand{\keywords}[1]
{
	\small	
	\textbf{\textit{Keywords---}} #1
}

\usepackage{amsmath, amssymb,multirow,array,amsthm}
\usepackage[svgnames]{xcolor}
\usepackage{tikz,multirow, array,xargs}
\usetikzlibrary{matrix}
\usetikzlibrary{backgrounds}
\usetikzlibrary{positioning}
\usetikzlibrary{shapes}

\theoremstyle{plain}
\newtheorem{theorem}{Theorem}
\theoremstyle{definition}
\newtheorem{definition}{Definition}
\newtheorem{proposition}{Proposition}

\begin{document}

\title{Searching for a source of difference  in graphical models
}


\author{Vera Djordjilovi\'c $^{1,\ast}$ \& Monica Chiogna$^{2}$\\[4pt]
	{$^{1}$Department of Economics, Ca' Foscari University of Venice, Italy }\\
	$^2$ Department of Statistical Sciences, University of Bologna, Italy}



\maketitle

\begin{abstract}
We look at a two-sample problem within the framework of decomposable graphical models. When the global hypothesis of equality of two distributions is rejected, the interest is usually in localizing the source of difference. Motivated by the idea that diseases can be seen as system perturbations, and by the need to  distinguish between the origin of perturbation and  components affected by the perturbation,  we introduce the concept of a \textit{minimal seed set}, and its graphical counterpart  a \textit{graphical seed set}. They intuitively consist of variables driving the difference between the two conditions.  
We propose  a simple testing procedure,  linear in the number of nodes, to estimate the graphical seed set from data. We illustrate our approach in the context of gene set analysis, where we show that is possible to zoom in on the origin of  perturbation in a gene network.

\end{abstract}

\section{Introduction}
\label{intro}
\subsection{Motivation}
The present work is motivated by the problem of identifying the origin of perturbation in gene regulatory networks.  In biological networks, diseases  can be modelled as  perturbations that affect certain targets, which, once perturbed, propagate the   perturbation through network connections \citep{del2010diseases}. In practice, we often collect and compare observations from healthy individuals and observations from patients after the  disease related perturbation has already taken place. On the basis of this comparison,  it is of interest to  identify the site of original perturbation, i.e., the {\it source of difference}, and distinguish it from the elements of the network that were affected through the process of network propagation.  
\subsection{Statement of the problem and some notation}

Let 
$\mathcal{F} = \{P_{\theta}; \theta \in \Theta\},$ $\Theta \subset R^d,$  be a family, parametrized by $\theta$, of probability distributions  for the random vector $X_V$, indexed by a set $V,$  $|V|=p,$ with support $\mathcal{X}_V.$ 
In what follows, to unburden the notation and when no ambiguity can arise, {we adopt the notation of \cite{dawid1993hyper} and, allowing for} a slight abuse of notation, we write $\theta$  instead of $P_\theta$ to denote individual distributions belonging to $\mathcal{F}$.    For $A,B \subseteq V,$ we will further  write $\theta_A$ to denote (the parameters of) the marginal distribution of variables in $A$  and, similarly, $\theta_{A|B}$  to denote a collection of conditional distributions $\left\{\theta_{A \mid X_B = y}, y \in \mathcal{X}_B\right\}$ indexed by $y$, where $X_B, \, B\subseteq V,$ is a subvector of $X_V$ and $\mathcal{X}_B$ is the associated  support. Different experimental conditions  will be distinguished by use of superscripts.

Consider a random vector $X_V \sim P_\theta$.  Within the context of two sample problems, the interest is often in testing the null hypothesis of equality of distributions $H_0: \theta^{(1)} = \theta^{(2)}$. If that hypothesis is rejected, one  usually aims at   localizing the source of difference. 

A common approach to tackle the question in genomics applications is to focus on the $p$ univariate marginal distributions, see for instance  \cite{ritchie2015limma} for a particularly popular method choice. 
Marginally speaking, a variable $X_v$, $v\in V$, can be  considered  relevant to the aim at hand if its marginal distribution  is  different in  $P_{\theta^{(1)}}$ and $P_{\theta^{(2)}}$. 

The (index) set of the relevant variables is then  taken to be
$$
R = \left\{v \in V:  \theta_v^{(1)} \neq \theta_v^{(2)}\right\}.
$$
Whether a variable belongs to $R$ depends solely on its  marginal distribution. 

Although   simple and computationally feasible,  the marginal approach might fail to point to the true source of difference whenever an interplay between variables plays a  role in differentiating the two distributions \citep{hudson2009differential}. In that case,  we propose to privilege a conditional perspective and exploit an approach which takes into account the entire $p$-dimensional joint distribution and flags a variable relevant only if the difference in its marginal distribution cannot be explained by the remaining variables. We define the set of conditionally relevant variables $D$ as follows.

\begin{definition} [Seed set]\label{seedsetdef}
	Consider  $\theta^{(1)},  \theta^{(2)} \in \mathcal{F}$. 	We call the set  $D\subseteq V$ the \textit{seed} set, if  
	the collections of conditional  laws $\theta_{V \setminus D\mid D}^{(1)}$ and $\theta_{V\setminus D\mid D}^{(2)}$ coincide.
	Furthermore, we say that $D$ is a \textit{minimal} seed set, if no proper subset of it is itself a seed set. 
\end{definition}

To facilitate the understanding of the above definition, it is helpful to consider that, by employing the factorization $p(x; \theta) = p(x_D; \theta_D) p(x_{\bar{D}} \mid x_{D}; \theta_{\bar{D}\mid D})$, 
where $\bar{D}= V \setminus D,$ the likelihood ratio $p(x; \theta^{(1)}) / p(x; \theta^{(2)})$  
simplifies to $p(x_D; \theta_D^{(1)}) / p(x_D; \theta_D^{(2)})$. The likelihood ratio thus depends only on  variables in $D$. When comparing the two distributions, the variables outside of $D$ are either irrelevant or redundant  and $D$ can be seen as the minimal subset of variables explaining the difference between the two distributions. It should be stressed that there is no relation between $R$ and $D$; in general neither $R \subseteq D$ nor $D \subseteq R$.

In practice, to identify the seed set, $D$ needs to be estimated from  data. One could perform a  number of tests of equality of conditional distributions, but
when $p$ is large, this testing problem becomes extremely challenging, and represents an open area of research, see for instance \cite{zhu2016two} and references therein. In this paper, we assume that the dependence structure among the $p$ variables in the joint distribution can be well represented by an undirected graph.
We then address the problem of identifying  $D$ within the framework of graphical models, where  we exploit the structural modularity of decomposable graphical models \citep{frydenberg1989decomposition,dawid1993hyper}. To this aim, we assume that $\mathcal{F}$ is a strong meta Markov model with respect to 
a given undirected  decomposable graph $G= (V,E)$, where $E\subseteq V \times V$ is a set of edges. 
Let us denote by $\mathcal{M}(G)$  a family of distributions satisfying the global Markov property relative to $G$. According to the definition introduced by \cite{dawid1993hyper},  $\mathcal{F}\subseteq \mathcal{M}(G)$   is a  strong meta Markov model  if for any decomposition ($A,B$) of $G$,  parameters $\theta_A$  and $\theta_{B\mid A}$ are variation  independent in $\mathcal{F}$ \citep[p.26]{barndorff2014information}. In other words, all possible  values of $\theta_A$ are logically compatible with all possible values of $\theta_{B\mid A}$.

Under this assumption, there is a close  relationship between the parametric model structure and the underlying graph, and 
we  show that the problem of identifying  $D$ can be formulated as  the problem of testing equality of  lower dimensional conditional distributions  induced by the structure of $G$. We further show that the associated test statistics   are functions of the  quantities pertaining to the lower dimensional marginal distributions. The key advantage is that inference on   marginal distributions is significantly less challenging than inference on conditional distributions.  Beside the computational gain, we argue that the proposed approach addresses the issue of exploiting information on the structure of dependence in an efficient and elegant way.  

\section{Decomposition of the global hypothesis of equality of  two Markov distributions}
\label{equality}

A major appeal of decomposable graphs in graphical modelling is that they allow for a clique-grained decomposition of the statistical model. 
Let $C_1,\ldots,C_k$ be a sequence  of  cliques of $G$ satisfying a running intersection property (see Section \ref{preliminaries} in Appendix), and let $S_2,\ldots, S_k$ be an associated sequence of (possibly non-unique) separators. Then, if the distribution of  $X_V$ is Markov relative to $G$, its joint  distribution decomposes as:
$$
p(x_V) = p(x_{C_1}) \prod_{j=2}^k p(x_{R_j}\mid x_{S_j}),
$$
where $R_j = C_j \setminus S_j$, $j=2,\ldots,k$. Therefore, each distribution $\theta \in  \mathcal{F}$ can be uniquely decomposed into $k$ lower dimensional components: $\theta_{C_1}, \theta_{R_2\mid S_2}, \ldots, \theta_{R_k\mid S_k}$;   uniqueness ensures that  $\theta$ can be reconstructed back from its components. As a consequence, the global hypothesis of equality $H: \theta^{(1)} = \theta^{(2)}$ also decomposes along the perfect ordering as $H= \cap_{j=1}^{k} H_j$, where $H_1: \theta_{C_1} ^{(1)} = \theta_{C_1}^{(2)}$ and $H_j: \theta_{R_j \mid S_j}^{(1)} = \theta_{R_j\mid S_j}^{(2)}$, $j=2,\ldots,k$.  Since $\mathcal{F}$ is a strong meta Markov model, the components of $\theta$  are variation independent  and  there are no logical relations among  the $H_j$.  The following result states that the log-likelihood ratio for $H$ decomposes analogously and that all component test statistics can be computed in clique-induced marginal models. 

\begin{theorem}   
	\label{t1}                     
	Let  $X_{V,1}^{(1)},\ldots, X_{V,n_1}^{(1)}$  and $X_{V,1}^{(2)},\ldots, X_{V,n_2}^{(2)}$ be two independent random samples from, respectively, ${\theta^{(1)}}$ and   ${\theta^{(2)}}$, $\theta^{(l)}\in \mathcal{F}$,  $l=1,2$, where $\mathcal{F}$ is strong meta Markov model relative to $G$.	$H: \theta^{(1)}= \theta^{(2)}$
	and its decomposition $H= \cap_{j=1}^{k} H_j$, where $H_1: \theta_{C_1} ^{(1)} = \theta_{C_1}^{(2)}$ and $H_j: \theta_{R_j \mid S_j}^{(1)} = \theta_{R_j\mid S_j}^{(2)}$, $j=2,\ldots,k$.
	Let $\lambda(V)$ denote the log likelihood ratio criterion  for testing $H$ against a general alternative and let $\lambda(A)$ denote the log likelihood ratio criterion  for testing equality of distributions  induced by $A \subseteq V$. The following equality holds
	\begin{equation}
		\lambda(V) = \lambda(C_1) + \sum_{j=2}^k \left\{\lambda(C_j)-\lambda(S_j)\right\},
		\label{lrdecomp}
	\end{equation}
	where $\left\{\lambda(C_j)-\lambda(S_j)\right\}$ represents the log likelihood ratio for testing $H_j$.
	Moreover,  the $k$ terms on the right hand side of~\eqref{lrdecomp} are asymptotically independent under the null hypothesis. 
\end{theorem}

\begin{proof}
	The joint density of any random sample of size $n$ from $\theta \in \mathcal{F}$ factorizes as
	\begin{equation}
		\label{fact}
		p(x_{(n)}; \theta) = p\left(x_{C_1, (n)}; \theta_{C_1}\right)\prod_{j=2}^k p(x_{R_j, (n)}\mid x_{S_j, (n)}; \theta_{R_j\mid S_j}),
	\end{equation}
	where $x_{(n)}$ stands for $x_1, \ldots, x_n$.  Each  component can be maximized separately to obtain  maximum likelihood estimates $\hat{\theta}_{C_1}$ and $\hat{\theta}_{R_j\mid S_j}$.  Note that maximum likelihood estimate  of $\theta_{C_1}$ is the same whether based on $x_{(n)}$ or $x_{C_1,(n)}$. 
	
	The likelihood ratio for testing $H$ is
	$$
	L(x_{(n_1+n_2)})= \frac{p\left(x_{(n_1+n_2)}; \hat \theta\right)}{p\left(x_{(n_1)}^{(1)}; \hat{\theta}^{(1)}\right)p\left(x_{(n_2)}^{(2)}; \hat{\theta}^{(2)}\right)},
	$$
	where $x_{(n_1+n_2)}$ denotes a pooled sample, $\hat{\theta}$ is the maximum likelihood estimate of $\theta^{(1)}=\theta^{(2)}$ under the null hypothesis, and $\hat{\theta}^{(l)}$, $l=1,2,$ is the maximum likelihood estimate of $\theta^{(l)}$ under the alternative. Factorizing each density as in \eqref{fact},  $L$ is decomposed into $k$ components corresponding to the local hypotheses $H_j, j=1,\ldots, k$. Using the equality { $\theta_{R_j\mid S_j}(x_{R_j}\mid x_{S_j})= \theta_{C_j}(x_{C_j})/ \theta_{S_j}(x_{S_j})$},
	we obtain the expression $\lambda(V) = \lambda(C_1) + \sum_{j=2}^k \left\{\lambda(C_j)-\lambda(S_j)\right\}$.
	Finally,  given the modular structure of the joint distribution, the  number of degrees  of
	freedom associated to $\lambda(V)$  is exactly the  sum of the degrees of freedom of the summands on the right hand-side, which is a sufficient  condition for the asymptotic independence of chi square random variables \citep{tan1977distribution}. 
\end{proof}

In what follows, we give explicit expressions for the decomposition for two important parametric families of distributions.  
\subsection{Gaussian graphical models}
Consider a subfamily of $\mathcal{M}(G)$ composed  of Gaussian graphical models.
In this case, $\theta = (\mu, \Sigma)$, with  $\mu \in \mathbb{R}^{p}$ and $\Sigma$ a symmetric positive definite matrix such that $\Sigma^{-1}\in S^+(G),$ where $S^+(G)$ denotes the set of all symmetric  $p \times p$ positive definite matrices  with zeros corresponding to the missing edges of $G$. For $A,B\subset V$, let $\Sigma_{AB}$  denote the corresponding submatrix of $\Sigma$ and let $\Sigma_A$ stand for $\Sigma_{AA}$.

For a given perfect clique ordering, the global hypothesis of equality $H: \theta^{(1)} = \theta^{(2)}$ decomposes as $H= \cap_{j=1}^k H_j$, with $H_1: \mu^{(1)}_{C_1} = \mu^{(2)}_{C_1}, \Sigma^{(1)}_{C_1} = \Sigma^{(2)}_{C_1}$ and $H_j: \theta_{R_j \mid S_j}^{(1)} = \theta_{R_j\mid S_j}^{(2)}$,  $j=2,\ldots,k,$ where
$$\theta_{A \mid B} = (\mu_{A} - \Sigma_{AB}\Sigma_{B}^{-1}\mu_{B}, \Sigma_{AB}\Sigma_{B}^{-1}, \Sigma_A - \Sigma_{AB}\Sigma_B^{-1}\Sigma_{BA}),$$
for $A,B \subset V,$ denotes parameters of the conditional law. 

Given two independent random samples of sizes $n_1$ and $n_2$ from $\theta^{(1)}$ and $\theta^{(2)}$, respectively, the log likelihood ratio $\lambda(A), A\subseteq V$, for testing the associated null hypothesis of equality is 
$$
\lambda(A)= \sum_{l=1}^{2}n_l\log\frac{|\hat{\Sigma}_A|}{|\hat{\Sigma}_A^{(l)}|},
$$
where $|\hat\Sigma|$ is determinant of  the maximum likelihood  estimate of $\Sigma$ under $H$, and $\hat{\Sigma}^{(l)}$, $l=1,2,$ are maximum likelihood estimates of $\Sigma^{(l)}$ under the general alternative \citep[p.416]{anderson2003introduction}. Since  $(\hat{\Sigma})^{-1}, (\hat{\Sigma}^{(l)})^{-1}\in S^+(G) $, $l=1,2$, and the determinant of every  $\Omega$ for which $\Omega^{-1} \in S^+(G)$ can be decomposed with respect to the graph as  $|{\Omega}|=\prod_{i=1}^k|\Omega_{C_i}|/\prod_{i=2}^k|{\Omega}_{S_i}|$ \citep[p.145]{lauritzen:96},  the log likelihood ratio $\lambda(V)$ can be equivalently written as
$
\lambda(V)= \sum_{i=1}^k\lambda(C_i)-\sum_{i=2}^k\lambda(S_i),
$
from which  equality of Theorem 1 follows. It is important to stress that when subgraph induced by $A$ is complete, which is the case with cliques $C_i$ and separators $S_i$,  then maximum likelihood estimate of $\Sigma_A$ is unconstrained.  {\color{black} In particular, if for ease of notation we temporarily drop the index $A$ in $x_A^{(l)},\,l=1,2$,  and write $x^{(l)}$ instead, we have 
$$
\hat{\Sigma}_A = \frac{1}{n_1+n_2}\left[\sum_{i=1}^{n_1}(x_i^{(1)} -\bar{x})(x_i^{(1)}-\bar{x})^T + \sum_{j=1}^{n_2}(x_j^{(2)} -\bar{x})(x_j^{(2)}-\bar{x})^T\right], 
$$
where $\bar{x}= (n_1\bar{x}_1 + n_2\bar{x}_2)/(n_1+n_2)$, whereas  $\hat{\Sigma}_A^{(1)}$ and $\hat{\Sigma}_A^{(2)}$ are unconstrained estimates of $\Sigma_A$ computed in the two samples, i.e. 
$$
\hat{\Sigma}_A^{(l)} = \frac{1}{n_l} \sum_{i=1}^{n_l} (x^{(l)}_i-\bar{x}_l)(x^{(l)}_i-\bar{x}_{l})^\top, \quad \bar{x}_{l} =\frac{1}{n_l}\sum_{i=1}^{n_l} x^{(l)}_i, \quad l=1,2.
$$
} In other words, it is possible to compute $\lambda(V)$ from test statistics computed in clique-induced marginal models in which maximum likelihood estimation is unconstrained.

\subsection{{Graphical log-linear models}}
Consider a subfamily $\mathcal{P}\subset \mathcal{M}(G)$ of graphical log-linear models. Each $X_v$ is now a categorical random variable with a finite set of possible values or levels $\mathcal{I}_v$. Here, $\mathcal{X}_V= \times_{v \in V} \mathcal{I}_v$.  We  refer to the elements of $\mathcal{X}_V$ as table cells \citep[Chapter 4]{lauritzen:96}). Let $X_{V,1}, \ldots, X_{V,n}$ be  $n \in \mathbb{N}$ independent realizations of  $X_V$.  Cell counts are defined as 
$$
n(h) = \sum_{i=1}^n I\left\{X_{V,i} = h\right\}, \quad h \in \mathcal{X}_V,
$$ 
where $I\left\{\cdot\right\}$ denotes the indicator function. 

For $A\subset V$, table cells $h_A \in \mathcal{I}_A = \times_{v\in A} \mathcal{I}_v$ are obtained by classifying observations only with respect to the variables in $A.$ Marginal cell counts  are $n(h_A)= \sum_{i=1}^{n} I\left\{X_{A,i} = h_A\right\}$, where $X_{A,i}$ is a subvector of $X_{V,i}$ induced by  $A$. 

Under a multinomial sampling scheme, the probability of the observed cell counts is
$$
\mathrm {Pr}(N(h)=n(h), h \in \mathcal{X}_V) = \frac{n!}{\prod_{h \in \mathcal{X}_V} n(h)!}\prod_{h \in \mathcal{X}_V}p(h)^{n(h)},
$$
where $p(h)$ is the probability for cell $h \in \mathcal{X}_V$.  In this case, $\theta = \left\{p(h)\right\}_{h \in  \mathcal{X}_V}$ satisfies the constraint  $\sum_{h \in  \mathcal{X}_V} p(h) =1$ and decomposes as $\theta_{C_1}= \left\{p(h_{C_1})\right\}_{h_{C_1}\in \mathcal{X}_{C_1}}$, which refers to the marginal table induced by $C_1$, and $\theta_{R_j\mid S_j}$ for $j=2,\ldots,k$, where $\theta_{A\mid B} = \left\{p(h_A \mid h_B)\right\}_{h_{A\cup B} \in \mathcal{X_{A \cup B}}}$ refers to the parameters of the $h_B$-slice of the table,  i.e.,  a table in which  objects are classified with respect to  the variables in $A$ for a given fixed level of the variables in $B$.

Consider now $\theta^{(1)},\theta^{(2)} \in \mathcal{P}$ and the null hypothesis of equality of probabilities in the marginal  table induced by $A\subseteq V$. Given two independent random samples with observed cell counts  $n^{(1)}$ and $n^{(2)}$ from $\theta^{(1)}$ and $\theta^{(2)}$, respectively, the log likelihood ratio $\lambda(A)$ is 
$$
\lambda(A)= 2\left\{\sum_{h_A\in \mathcal{X}_A}\sum_{l=1}^2 n^{(l)}(h_A)\log\left( \frac{\hat{p}^{(l)}(h_A)}{\hat{p}(h_A)}\right)\right\},
$$
where $\hat{p}_A$ is the maximum likelihood estimate of $p_A$ under the null hypothesis; and $\hat{p}^{(1)}_A$ and $ \hat{p}^{(2)}_A$ are maximum likelihood estimates of  $p^{(1)}_A$ and $p^{(2)}_A$ under a general alternative.  Using  the structural decomposition reflected in the maximum likelihood estimator $\hat{p}$:
\begin{equation}\label{mledecompct}
	\hat{p}(h)= \frac{\prod_{j=1}^{k} \hat{p}(h_{C_j})}{\prod_{j=2}^k \hat{p}(h_{S_j})}, \quad h\in \mathcal{X}_V,
\end{equation}
we obtain the decompisiton of $\lambda(V)$ featured in Theorem \ref{t1}. Degrees of freedom associated to $\lambda(V)$ can be computed from the formula $f(V) = f(C_1) + \sum_{j=2}^k \left\{f(C_j)- f(S_j)\right\}$, where $f(A)$ denotes degrees of freedom in a  model induced by $A\subseteq V$. Since marginal models induced by cliques and separators are saturated, their degrees of freedom are obtained as $f(C_j)= \prod_{v \in C_j} |\mathcal{I}_v| -1$, and analogously for separators.

\section{Estimation} 
\label{keyresult}
\subsection{The  graphical seed set}
\label{gssSection}

Before we show how the result of the previous section can be used to make inference about the seed set, we need to  introduce the concept  of the  graphical seed set. Namely,   by employing a clique-grained decomposition, we are not always able to identify the minimal seed set; in those cases we can identify its superset that we denote  by $D_G$.   Relation between the two sets, that depends on  both $D$ and $G$, is the subject of this section.

\begin{definition}[Graphical seed set] \label{gssdef}
	Let $D$ be a minimal seed set for   $\theta^{(1)}$ and $\theta^{(2)} $, two  graphical distributions  Markov with respect to     $G$. Let  $\mathcal{S}=\left\{S: S \mbox{ is a separator in } G\right\}$ be   the collection of  separators in $G$. Then  we call the set  
	
	\begin{equation}
		D_G =\left\{v \in V \mid \forall S\in \mathcal{S},\mbox { either } v\in S \mbox{ or } S \mbox{ does not separate } v \mbox{ from  } D  \mbox{ in } G\right\}
		\label{gssD}
	\end{equation}
	a \textit{graphical seed set}.
\end{definition}

In the above definition, we allow for non-empty intersection between $S$ and $D$, as well as $S=D$.
When $v\in D$, the condition \eqref{gssD} is trivially satisfied ($v$ cannot be separated from $D$ by any set), and therefore $D_G \supseteq D$. 	The graphical seed set $D_G$ is thus the smallest  set containing the seed set $D$ that can be identified by means of  set operations on cliques and separators of $G$. 

When the minimal seed set is a separator, we can set $S=D$ in \eqref{gssD}, to obtain $D=D_G$. In general,  $D$ and $D_G$ will coincide whenever $D$ can be expressed as an intersection of two or more cliques. In other instances,  $D_G$ will be a seed set, but not a minimal one. {For an illustrative example, see Section \ref{example} in Appendix. }

\subsection{The graphical seed set estimator}
We have seen above that the global hypothesis of equality can be decomposed according to a specified perfect ordering  into a set of local hypotheses. 
However, the perfect ordering is not unique. In fact, there are multiple decompositions of the global hypothesis, each   corresponding  to a different factorization of the same  distribution.  It is this multiplicity that we  exploit when estimating  the graphical seed set. 

For a given graph, the enumeration of all decompositions might resemble the problem of enumerating  its junction trees \citep{thomas2009enumerating}, but a closer look reveals that it is a far simpler  task. Given the uniqueness of the sequence of separators, it is not difficult to show that there is exactly one decomposition for  each choice of the  root clique -- the clique labeled $C_1$ -- leading to a total of  $k$  decompositions.

Before we show how these  different decompositions relate to the graphical seed set in Proposition~\ref{gss}, we introduce some notation and  restate the global testing problem in decision theory terms.
Let { $\Theta \times \Theta$} be the unrestricted parameter space of { $(\theta^{(1)},\theta^{(2)});$ let $\Theta_0 = \left\{ (\theta, \theta); \theta \in \Theta\right\}$ denote the space restricted by  $H: \theta^{(1)}=\theta^{(2)}$, and  let  $\Theta_1 = (\Theta \times \Theta) \setminus\Theta_0$. We want to test $H: (\theta^{(1)}, \theta^{(2)}) \in \Theta_0$ against a general alternative $ (\theta^{(1)}, \theta^{(2)}) \in \Theta_1.$} Let the decision taken on $H$ be denoted by  $d$, where $d=0$ means that the null hypothesis is not rejected and $d=1$ means that the null hypothesis is rejected. A test  $\phi$ is a mapping  from the sample space to the set $\left\{ 0, 1\right\}$ (we rule out the trivial case that the test  makes no decisions). 
Let $d^*$ denote the correct decision (the truth)  for  $H$. 
As seen in the previous Section, the null hypothesis can be decomposed into a set of independent local hypotheses,
i.e., $H = \bigcap_{j=1}^{k} H_{j}$,
and we denote by  $d^*_j$  the correct decision for $H_{j}, \, j=1,  \ldots k$, so that   $d^*= (d^*_1, \ldots, d^*_k). $ To identify the $i-$th decomposition, obtained when   $C_i$ is set as the root clique, we let $C_{i,1},\ldots,C_{i,k}$ denote a sequence of cliques satisfying the running intersection property. Let $S_{i,2},\ldots,S_{i,k}$ be an associated sequence of separators, and set $S_{i,1}=\varnothing$, $i=1,\ldots,k$. In this notation, $H_{i,j}$ will denote the $j-$th null hypothesis in decomposition $i$, $\phi_{i,j}$ the corresponding test,  and $d^*_{i,j}$ the associated correct decision. 

We now show the connection between the graphical seed set and  the decompositions obtained from the graph $G$. 
\begin{proposition} \label{gss}
	Let $d_i^*=\left(d_{i,1}^*,\ldots,d_{i,k}^*\right)$ be the vector of correct decisions for the hypotheses $H_{i,j}$ of equality of collections of conditional  distributions of  $X_{R_{i,j}} \mid X_{S_{i,j}}$ in the $i-$th decomposition.  Then
	\begin{equation*}
		D_G=\bigcap_{_{i=1}}^{k}\bigcup_{\left\{j: \,\,d_{i,j}^*=1\right\}} C_{i,j}.
		\label{procedure}
	\end{equation*}
\end{proposition}

The above proposition gives an oracle procedure for recovering the graphical seed set from the knowledge of the two joint distributions. In practice, we need to rely on statistical tests.  Let $\phi_i = \left(\phi_{i,1}, \ldots, \phi_{i,k}\right)\in \left\{0,1\right\}^k$ be a vector indicating the results  of the statistical tests performed in the $i$-th decomposition, $i=1, \ldots k,$ with $\phi_{i,j}=1$ when the hypothesis $H_{i,j}$ is rejected, and  $\phi_{i,j}=0$ otherwise. The following definition naturally follows.

\begin{definition}[Graphical seed set estimator]
	The random set $\hat{D}_G$, defined as
	\begin{equation}
		\hat{D}_G = \bigcap_{_{i=1}}^{k} \bigcup_{\left\{j: \,\,\phi_{i,j}=1\right\}} C_{i,j}
		\label{graphicalseedset}
	\end{equation}
	is an estimator of  $D_G.$  
\end{definition}

\subsection{Asymptotic behavior}
Estimator $\hat{D}_G$ is  different from classical estimators in that its values depend on 
data through the results of  sequences of tests. Properties of the estimator will ultimately
depend on the  properties of the tests which are used.  A treatment of these properties in the limit of infinite data benefits from the introduction of a more general notion of consistency of tests, that we give in general terms as follows (see Definition 1 in \cite{robins2003uniform} for a similar treatment).
\begin{definition}
	A  sequence of tests $\phi(n)$ for the hypothesis $H: (\theta^{(1)}, \theta^{(2)}) \in \Theta_0$ vs $H_1: (\theta^{(1)}, \theta^{(2)}) \in \Theta_1$ is consistent if  for each $(\theta^{(1)}, \theta^{(2)}) \in \Theta \times \Theta$ there exists 
	a sequence  of significance levels $\alpha_n$ s.t. 
	\begin{itemize}
		\item[(1)] for each $(\theta^{(1)}, \theta^{(2)}) \in \Theta_0,\,\,\,$ $ \lim_{n\rightarrow \infty} \mathbb{P}_{(\theta^{(1)}, \theta^{(2)})}(\phi(n) = 1) = 0;$
		\item[(2)] for each $(\theta^{(1)}, \theta^{(2)}) \in \Theta_1,\,\,\,$ $ \lim_{n\rightarrow \infty} \mathbb{P}_{(\theta^{(1)}, \theta^{(2)})}(\phi(n) = 0) = 0.$
	\end{itemize}
\end{definition}
In other words, a sequence of tests is consistent if, at least asymptotically, it  reports a correct decision.
Let us now consider testing $H_{i,j}$ in the above given framework. Let $n = n_1 + n_2$ and assume that as $n \rightarrow \infty,$ $n_l/n \rightarrow \gamma_l$ such that $0< \gamma_l< 1, \,\, l=1,2,$ and $\gamma_1+\gamma_2=1.$  Moreover, let the test statistic $\phi_{i,j}(n)$ be defined as 
$$
\phi_{i,j}(n) = \begin{cases}
	0 & \lambda_{i,j;n} < q_n \cr
	1 & \lambda_{i,j;n} > q_n \cr 
\end{cases}
$$
where $\lambda_{i,j;n}$ is the log likelihood ratio  for $H_{i,j}$ and $q_n$ a suitable sequence of quantiles. 
Standard results assure that, under the null hypothesis, the sequence  $\lambda_{i,j;n}$ converges to a chi-square distribution with $f$ degrees of freedom, where $f$ is the difference between the  dimensions of the unrestricted  parameter space  and the restricted
parameter space implied by the hypothesis of equality of the distributions of  $X_{R_{i,j}} \mid X_{S_{i,j}}$
in the two groups. Then, the test that rejects the null hypothesis  if $\lambda_{i,j;n}$ exceeds 
the upper $\alpha$-quantile of the chi-square distribution is asymptotically of level $\alpha.$ 
We can state the following proposition.
\begin{proposition}\label{prop2}
	In the framework stated above, for each $H_{i,j}$, there exists a sequence of significance levels $\alpha_n$, s.t. the  sequence of tests $\phi_{i,j}(n)$ is consistent.
\end{proposition}

\begin{theorem}\label{t2}
	The estimator $\hat{D}_G$ is a pointwise consistent estimator of $D_G$,  i.e.,  $\mathbb{P}_{(\theta^{(1)},\theta^{(2)})}(\hat{D}_G = D_G) \rightarrow 1.$
\end{theorem}

\subsection{{\color{black}Finite sample   type I error control}}\label{fwerg}
With finite samples, it is customary to assign a bound to the probability of incorrectly rejecting the null hypothesis by imposing conditions such as $\mathbb{P}_{(\theta^{(1)}, \theta^{(2)}) \in \Theta_0}(\phi_{i,j}(n) = 1)\le \alpha.$ Estimation of  ${D}_G$ requires performing a collection of   $k+\sum_{i=1}^{k}\nu(C_i)$ tests, where $\nu(C_i)$ denotes the number of separators contained within the clique $C_i$.  Finite sample behavior of $\hat{D}_G$ thus  hinges on the proper control of the multiplicity issue. 

{\color{black}We focus on the requirement that  the probability that $\hat{D}_G$ contains a false positive should be bounded by a given $\alpha\in (0,1)$, i.e. 
${\mathrm P}(\exists v \in V: v\in \hat{D}_G \cap v\notin D_G) \leq \alpha$. But, if there is such a node $v$, then given Definition 3 of $\hat{D}_G$, necessarily one of the true null hypotheses in the collection of hypotheses $\mathcal{H} = \left\{H_{ij}, i,j = 1,\dots, k\right\}$ was erroneously rejected. This implies that the control of familywise error rate for $\mathcal{H}$, i.e. the probability of rejecting at least one true null hypothesis, results in the control of probability of including a false positive in $\hat{D}_G$.

The simplest approach to  control the familywise error rate  is to apply the Bonferroni correction with a factor of $k+\sum_{i=1}^{k}\nu(C_i)$. }However, the Bonferroni  correction can be overly conservative {\color{black} when there is high dependence among $p$-values. This is the case here, since although local test statistics are independent within a single decomposition (see Theorem \ref{t1}), considering  alternative decompositions leads to logical relations among hypotheses and typically results in a high positive dependence between  the associated  $p$-values.  }
To address this issue, we employ  the max$T$ method of \cite{westfall1993resampling}, which uses permutations to obtain the joint distribution of the $p$-values and, by accounting  for the dependence among $p$-values,  attenuates the conservativeness of the Bonferroni procedure. In our setting,  the condition of subset pivotality is satisfied, and  the  Westfall and Young procedure  controls the familywise error rate in the strong sense. 

In many applications, familywise error rate control  is considered too stringent and false discovery rate is considered instead. {\color{black} Unfortunately, no such simple relation exists between controlling false discovery rate for $\mathcal{H}$ and the inclusion of false positives in $\hat{D}_G$.  In other words, it is unclear how controlling false discovery rate for $\mathcal{H}$ translates to the type I error guarantees for $\hat{D}_G$.} For this reason, we restricted our attention to the familywise error rate.

\section{Simulation studies}
\label{sstudies}
\subsection{Simulation study 1}
\label{simstudy}
To study the finite sample behavior of $\hat{D}_G$, we considered a randomly generated graph $G$ consisting of 100 nodes grouped in 37  cliques (the largest clique containing 15 nodes).   The code to reproduce all numerical experiments, as well as real data analysis featured in Section \ref{rd}, is available at \texttt{https://github.com/veradjordjilovic/Seed-set}.
A plot of the graph is shown in Figure \ref{fig:SimulationStudy} in Appendix. The minimal seed set was set to $D=\left\{2,5\right\}.$ In the chosen graph, the graphical seed set does not coincide with the minimal seed set since there is no separator in $G$ that separates a node number 17  from $D$. We thus have $D_G=\left\{2,5,17\right\}$.

We will work in the Gaussian setting. We set the parameters of the first, i.e.  control, condition in the following way. The means of 100  variables were drawn randomly  from a normal distribution centered at $0.5$ (standard deviation 1). The  covariance matrix was obtained by starting from a matrix with  all off-diagonal elements equal to 0.4 and all diagonal elements equal to 1 and modifying it so that its inverse has zeros corresponding to the missing edges of $G$. 
For the second  or the perturbed  condition, we considered {\color{black} perturbations that alter the means of the two seed set variables linearly. In particular, the means were multiplied by $\lambda$ that varied in the range $\left\{1.2, 1.25, \ldots, 1.6, 1.65\right\}.$}  
{\color{black}The variance of seed set variables was also manipulated and decreased by 50\%.   We held the sample size fixed and equal for the two conditions: $n_1 = n_2 = 50$. For each $\lambda,$ we generated 1000 pairs of samples. }

{\color{black} Note that this perturbation affecting $X_2$ and $X_5$, indirectly affected  all the marginal distributions of $(X_1, \ldots,X_{100})^\top$. For an illustration of this effect, see Figure~\ref{vcmat},  Appendix, that compares the parameters associated to the first ten variables, i.e., $X_1,\ldots, X_{10},$  in the first and in the second condition for $\lambda = 1.7$.}

We computed  $\hat{D}_G$ with the \texttt{SourceSet}  \texttt{R} package, which implements the proposed approach (available from CRAN).  The familywise error rate was controlled  at 5\% by  the step-down max$T$ method  \citep{westfall1993resampling}. To evaluate the performance of our procedure, we  computed  {\color{black} the empirical power, defined as the frequency with which  the estimated graphical seed set $\hat{D}_G$ coincided with the true graphical seed set $D_G$, and the empirical familywise error rate, defined as the frequency with which $\hat{D}_G$ contained a false positive. 
The results are shown in Figure \ref{ss1a}.  
}  

\begin{figure}\centering
\includegraphics[width = 0.8\linewidth]{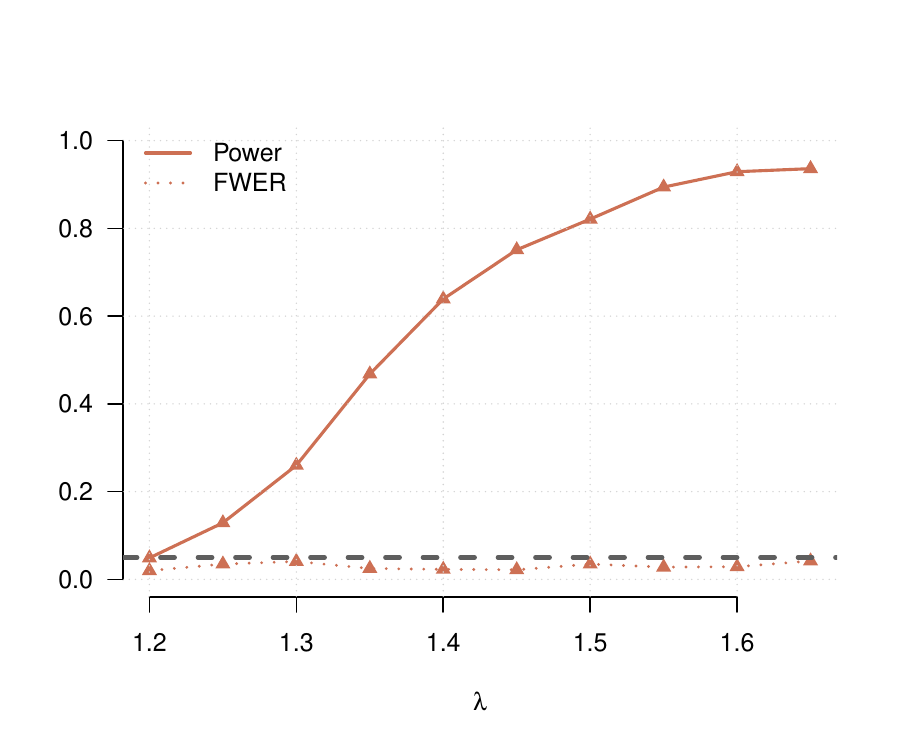}
\caption{\label{ss1a}Simulation study 1: Empirical power and familywise error rate of the graphical seed set estimating procedure as a function of perturbation strength $\lambda$.  Dashed horizontal line $y = 0.05$, representing the nominal familywise error rate, was added for reference.}
\end{figure}


{\color{black} Results show that the familywise error rate is controlled at the nominal level for all values $\lambda$, which is in line with finite sample theoretical type I error guarantees described in Section \ref{fwerg}. With regards to power, for the lowest level of perturbation $\lambda = 1.2$, corresponding to an increase of 20\% in variables $X_2$ and $X_5$, we see that the power to identify $D_G$ is very low. With increasing $\lambda$, the power is fast increasing and reaches $80\%$ already for $\lambda=1.5$. Note that given our definition of power, the maximum attainable power is  bounded by the complement of the familywise error rate, i.e.  $1- \mathrm{P}(\exists v \in V: v \in \hat{D}_G \land v \notin D_G) \approx 1-\alpha$, rather than 1. }

{\color{black}
\textbf{Unbalanced sample sizes.}
We further studied the impact an unbalanced sample size can have on the performance of the seed set estimating procedure. To this end, we fixed parameters of  the perturbed condition by setting $\lambda = 1.3$ and then varied the sample size of the pooled sample $n = n_1 + n_2$ in the set $\left\{75, 100, 125, 150, 200, 250, 300, 350\right\}$. We computed the empirical power and familywise error rate in two scenarios featuring:
\begin{itemize}
\item balanced samples: $n_1=n_2$ when $n$ is even, or $n_1 = \left \lfloor{n/2}\right\rfloor$ and $n_2 = n_1+1$, when $n$ is odd;
\item unbalanced samples: $n_1 = 50$ and $n_2 = n- n_1$. 
\end{itemize}
Results, shown in Figure \ref{ss1b}, indicate that the familywise error rate is controlled well in both scenarios. With regards to power, when the total sample size is small, the two scenarios are comparable. With increasing sample size, the difference between $n_1$ and $n_2$ is also increasing, and the power in the scenario with balanced samples is higher, but the advantage does not seem to be very large. 
\begin{figure}\centering
\includegraphics[width = 0.8\linewidth]{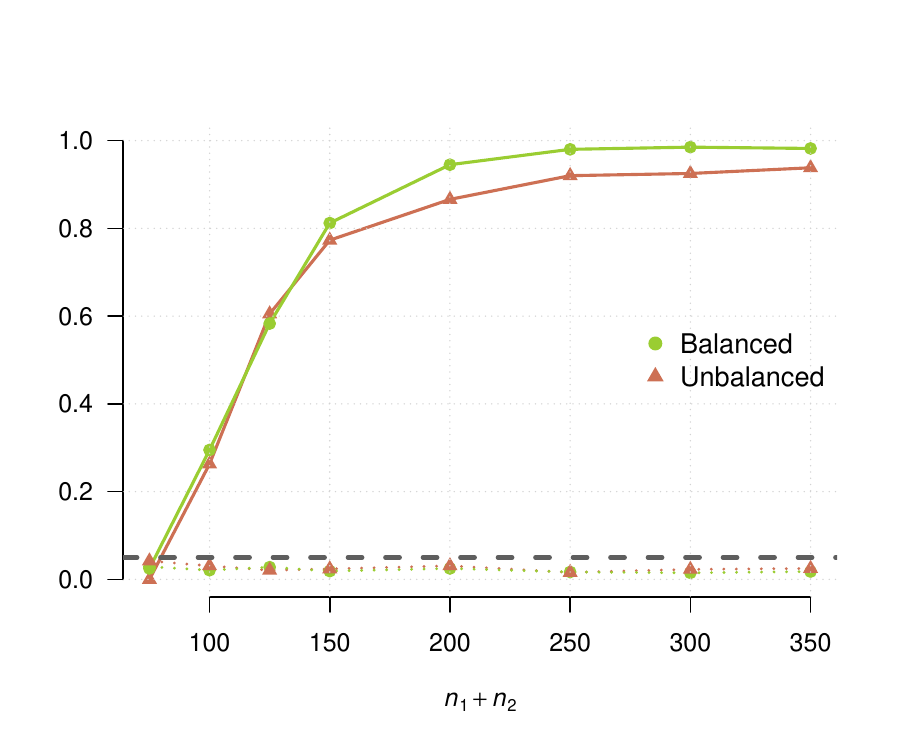}
\caption{\label{ss1b}Simulation study 1: Empirical power and familywise error rate of the graphical seed set estimating procedure as a function of the pooled sample size $n = n_1+n_2$.  In an unbalanced sampling scheme $n_1=50$ was fixed, while in a balanced sampling scheme $n_1=n_2$ if $n_1+n_2$ was even, and $|n_1-n_2| =1$ otherwise. Dashed horizontal line $y = 0.05$, representing the nominal familywise error rate, was added for reference.}
\end{figure}

\textbf{Robustness to non-normality. } An important issue arising in practical applications is the sensitivity of the procedure to  the presence of departures from normality. To investigate this issue, we have considered data generated from skew-normal graphical models \citep{capitanio2003graphical} and studied the power and familywise error rate as a function of skewness. The results of this simulation study, described in Section~\ref{SMSkew}, Appendix, suggest that when compared to a setting with normal data,  the power does not seem to be much affected, while the  familywise error rate increases and possibly surpasses the pre-specified level $\alpha$. Nevertheless, the increase seems to be small enough as to allow us to conclude that the procedure is quite robust to this particular violation of normality. 

\textbf{Competing methods.} To the best of our knowledge,  no alternative methods aiming to estimate $D_G$,  i.e. the origin of the perturbation affecting both the means and the (co)variances are currently available. However, some recent approaches focus on detecting more specific forms of perturbations: either those affecting exclusively the graphical structure or the vector of means.  In the following section, we report the comparison with a method addressing the former, while in Section~\ref{SMcomparison}, Appendix, we provide a comparison with a method addressing the latter.}

\subsection{Simulation study 2}
To study the behavior of our procedure  when the the difference between two conditions is driven only by the graphical structure, we considered a small graph consisting of 10 nodes, shown in Figure \ref{differentstructure}.  The edge between nodes 4 and 6 is present in condition 1, but absent in condition 2,   i.e.,  in condition 2, variables associated to nodes 4 and 6 are conditionally independent given the rest. It is worth noting that, in condition 2, the graph is  not decomposable and that the graphical structure to be used in estimating $D$ is that of condition 1, as it represents the decomposable model common to the two conditions. The minimal seed set is now $D=\left\{4,6\right\}$, and it  coincides with the graphical seed set. 

Means of the 10 variables were randomly drawn from  a normal distribution centered at $0.5$ (standard deviation 1) and were the same for conditions 1 and 2.   In each condition, the  covariance matrix  was obtained from a  matrix with  all diagonal elements equal to 1 and all off-diagonal elements equal to 0.6, that was modified so that the zero pattern of its inverse corresponds to the missing edges of $G$.  Three different sample sizes were considered, i.e., $n=200, 300, 500$.

Results, averaged over 500 Monte Carlo runs, are shown in Table \ref{tablelabel}, where rows labeled `Seed set' report the percentage of times each node was found to belong to $D$.  Results show that, in this setting, the power, although limited at the smallest sample size,  is increasing with increasing sample sizes. This is understandable, since, differently from simulation 1, the difference between the two conditions is relatively sparse, and the smaller this difference, the harder it is to distinguish between the null and the alternative hypothesis.

It is interesting noting that methods for differential networks,  such as those in \cite{zhao2014direct} and \cite{xia2015testing}, could also have been used in this setting. For an appreciation of the different results produced by different approaches,  we considered the  method of \cite{zhao2014direct}, for which an implementation is available. The method focuses only on the structure of the covariance; it uses no external information on such structure and it has been developed around estimation consistency.  It follows that this method is not directly comparable with our method, and its relative performance is to be interpreted with caution.

The implementation of the differential network method was obtained from the github account of the corresponding author of \cite{zhao2014direct}.  Cross validation and $L_\infty$ were chosen as  tuning criteria. The output of this method is  an estimate of the difference between two precision matrices. To facilitate comparison with our method, we focused on  the differential network given by a subset of non zero elements of the estimated difference.  A variable was deemed important if the associated node belonged to the estimated differential network, i.e. if at least one edge of the differential network featured  the  node in question.   In this case, the true differential network consists of a single edge joining nodes 4 and 6.  Variables deemed important by this method should thus coincide with the minimal seed set.  

Rows labeled `Differential network'  in Table \ref{tablelabel}, report the percentage of times a variable belonged to the set of important variables according to the differential network method. The method flags nodes 4 and 6 to be relevant  also for the smallest sample size (around 85\% of times for $n=200)$.  However, the rate of a false discovery is much higher, around 40\% across the remaining nodes, and does not seem to be decreasing with increasing sample size. Note that this is not in conflict with the consistency of the estimator of  \cite{zhao2014direct}, since the estimated non-zero elements are getting smaller in absolute value (results not reported here) and converge to zero with increasing sample size. 

\begin{figure}
\centering
	\includegraphics[width = 0.4\linewidth]{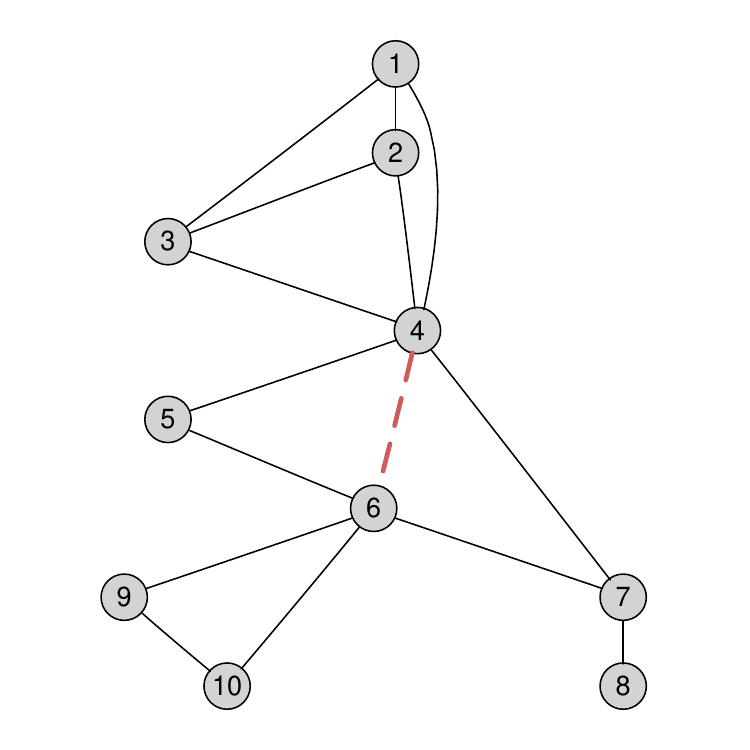}
	\caption{An undirected graph used in  Simulation study 2. Edge (4,6) is present in condition 1, and absent in condition 2.  }
	\label{differentstructure}
\end{figure}

\begin{table}
	\def~{\hphantom{0}}
	\caption{	\label{tablelabel}Simulation study 2: percentage of times (\%) a node is found to belong to $D$ or a differential network.Monte Carlo standard error of estimates is bounded by $2.2\%$.}{%
		\begin{tabular}{rlrrrrrrrrrr}
			&	&\multicolumn{10}{c}{Node}\\[5pt]
			&     & 1& 2 & 3 & 4 & 5 & 6 & 7 & 8& 9& 10\\[7pt]
			\multirow{ 2}{*}{$n=200$}	&		 Seed set& $1$ & $1$& $1$& $22$& $1$& $25$& $3$& $1$&$1$ &$1$\\
			& Differential network  &$34$ &$39$&$40$&$86$&$44$&$85$&$51$&$35$&$40$&$37$\\[0.5cm]
			\multirow{ 2}{*}{$n=300$}&		 Seed set&$1$ &$1$&$1$& $46$&$1$&$47$&$0$&$0$&$1$ &$1$\\
			& Differential network  &39 &37&37&93&51&94& 50&36&44&42\\[0.5cm]
			\multirow{ 2}{*}{$n=500$}&		 Seed set&$2$ &$2$&$2$&$86$&$2$&$86$&$2$&$1$&$0$&$0$ \\
			& Differential network  & 42&42&40&99&50&99&56&36&46&46			
	\end{tabular}}
\end{table}

\section{Biological validation}
\label{rd}
\label{realdata1}

Genes and gene products cluster into functionally connected pathways, i.e. networks of biological interactions that describe their basic dynamics \citep{kanehisa:2000}. A large literature has developed around the problem of detecting statistically significant dysregulations of pathways in different experimental conditions \citep{goeman2004global, hummel2008globalancova,tsai2009multivariate},  but translating  detected dysregulations into claims about their origin is a challenging task. Chromosomal  rearrangements offer a possible explanation.
Chromosome rearrangements  initiate various alterations of the regulation of gene expression through a variety of different mechanisms.  For this reason, when comparing populations with and without  a given gene rearrangement, sound inferential tools usually flag most pathways including genes with the rearrangement as statistically different. What we should expect from tools calibrated to detect the source of dysregulation is that they go as close as possible to the rearranged genes. This is the reason why we consider known chromosomal rearrangements as ideal case studies to explore the power of our  procedure on real, complex and noisy data.

As an example, consider the BCR/ABL fusion gene, formed by rearrangement of the breakpoint cluster region (BCR) on 
chromosome~22 with the c-ABL proto-oncogene on chromosome~9. This rearrangement has been postulated to be responsible for the development of leukemia
and is present in  all chronic myelogenous leukemia patients. It is also identified in some cases of 
acute lymphocytic leukemia (ALL), in which it is associated with poor prognosis.

We consider a well-known   dataset \citep{chiaretti2005gene} 
available  from an \texttt{R} package \texttt{ALL}\citep{ALL}. Data refer to gene expression signatures of two groups of ALL patients: a first group  of  
37 subjects with BCR/ABL gene rearrangement, and a second group of 41 subjects without the BCR/ABL gene 
rearrangement. In what follows, we will consider the  Chronic myeloid leukemia pathway, shown in Figure \ref{fig:ALLpathway} in Appendix,  a pathway whose functioning is highly impacted 
by BCR and ABL genes.

To derive the underlying undirected graph, we used the \texttt{R} package \texttt{graphite} 
\citep{graphite}, which transforms KEGG pathways into  graph objects.  We moralized and triangulated this graph to  obtain 
a decomposable graph. For  graph operations, we relied on the package \texttt{gRbase} \citep{grbase}.  The obtained graph 
consists of three connected components, and for illustration purposes, we restricted our attention to the largest connected 
component,  consisting of 27 nodes and 16 cliques, shown in Figure \ref{fig:ResultsChimera} (colors can be ignored for now). 
The number associated to each node is a unique  gene identifier from the Entrez Gene database at the National Center for 
Biotechnology Information \cite{maglott2005entrez}. Note that nodes 25 and 613 represent, ABL and BCR genes, respectively.

The global hypothesis of equality of distributions in the two groups is rejected by the 
likelihood ratio test ($p$\,-value $=2.06\times 10^{-11}$). 
To estimate $\hat{D}_G$, we  decomposed the 
graph into a succession of cliques. 
There are 16 cliques, and thus 16 decompositions of the global null hypothesis,  and  41 unique local hypotheses.  We controlled the familywise error rate at $5\%$ level  by the min$P$ 
method with $B=1640$ permutations (the minimal number recommended by the \texttt{SourceSet} package).  We have thus relied on permutation, rather than asymptotic $p$-values.    Obtained $p$-values are shown in Table \ref{tab_chimera2}. The threshold found by min$P$ method was $2.4\times 10^{-3}$. The resulting estimate   is represented  in Figure \ref{fig:ResultsChimera}. Highlighted nodes 
(either gray or red) belong to cliques that result significantly different in two conditions, while the red nodes 
form the estimated graphical seed set $\hat{D}_G=\left\{25,613,6776\right\}$. These three genes, thus, explain the marked difference  between the two groups, but their effect does not seem to propagate towards other genes in the network (the majority of white nodes in Figure \ref{fig:ResultsChimera}). 

\begin{figure}
	\centering
	\includegraphics[width=0.6\textwidth]{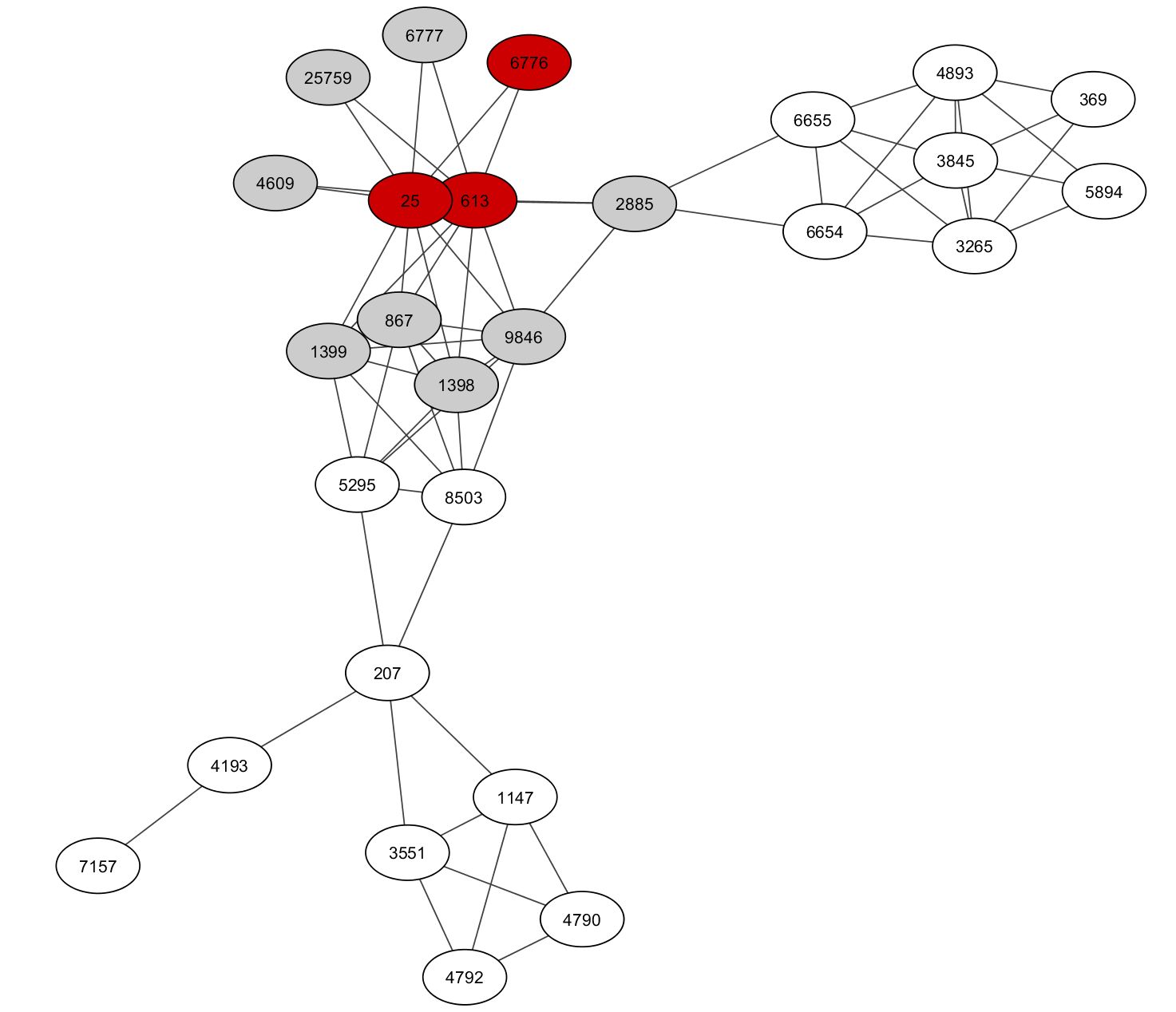}
	\caption{An undirected graph representing  the Chronic myeloid leukemia pathway. Genes belonging to cliques for which the hypothesis of equality of distributions is rejected are highlighted. Genes  belonging to the estimated graphical seed set are colored red. }
	\label{fig:ResultsChimera}
\end{figure}

\section{Discussion}
\label{discussion}
{
	Two sample testing problem we consider is closely related to the problem of variable selection in a logistic regression.  When a predictor is a $p$-dimensional random vector $X$ and the output is a class label (1 or 2),   the minimal seed set coincides with the Markov blanket of the response.

	Modularity of graphical modes is usually  considered with regards to  density factorization or parameter estimation. Theorem 1 mirrors this property in the hypothesis testing setting within the framework of strong meta Markov models, and although conceptually simple, we were unable to find this result in the literature.  The strong meta Markov assumption  is a strong assumption, however,  the two families most often encountered in practical applications,  that of Gaussian graphical models and graphical log-linear models, fall within this framework. 
	
	The presented approach estimates the graphical seed set which might be larger than the minimal seed set. An open question regards a potential  two-step procedure, in which clique grained decomposition is followed by additional tests aiming at identifying $\hat{D} \subseteq \hat{D}_G$.  Statistical properties of such a procedure are far from trivial, and  we leave this question for future research.
	
	Our approach is based on the assumption that the graphical structure is known, either derived from relevant subject matter considerations or estimated  from previous studies.  When this is not the case,  finding ways to combine learning of the graphical structure with the presented approach in an efficient way, while controlling the desired error rate, represents a methodological challenge  that  awaits further research.


}

\bibliographystyle{apalike}      
\bibliography{bibHyper_revision.bib}  

\begin{thebibliography}{}

\bibitem[Anderson, 2003]{anderson2003introduction}
Anderson, T.~W. (2003).
\newblock {\em An introduction to multivariate statistical analysis}.
\newblock Wiley, New Jersey.

\bibitem[Azzalini, 2021]{skewnormalR}
Azzalini, A. (2021).
\newblock {\em The {R} package \texttt{sn}: The Skew-Normal and Related
  Distributions such as the Skew-$t$ and the SUN (version 2.0.0).}
\newblock Universit\`a di Padova, Italia.

\bibitem[Azzalini and Capitanio, 1999]{azzalini1999statistical}
Azzalini, A. and Capitanio, A. (1999).
\newblock Statistical applications of the multivariate skew normal
  distribution.
\newblock {\em Journal of the Royal Statistical Society: Series B (Statistical
  Methodology)}, 61(3):579--602.

\bibitem[Barndorff-Nielsen, 2014]{barndorff2014information}
Barndorff-Nielsen, O. (2014).
\newblock {\em Information and exponential families in statistical theory}.
\newblock John Wiley \& Sons, New York.

\bibitem[Capitanio et~al., 2003]{capitanio2003graphical}
Capitanio, A., Azzalini, A., and Stanghellini, E. (2003).
\newblock Graphical models for skew-normal variates.
\newblock {\em Scandinavian Journal of Statistics}, 30(1):129--144.

\bibitem[Chiaretti et~al., 2005]{chiaretti2005gene}
Chiaretti, S., Li, X., Gentleman, R., Vitale, A., Wang, K.~S., Mandelli, F.,
  Foa, R., and Ritz, J. (2005).
\newblock Gene expression profiles of {B}-lineage adult acute lymphocytic
  leukemia reveal genetic patterns that identify lineage derivation and
  distinct mechanisms of transformation.
\newblock {\em Clinical Cancer Research}, 11(20):7209--7219.

\bibitem[Dawid and Lauritzen, 1993]{dawid1993hyper}
Dawid, A. and Lauritzen, S. (1993).
\newblock Hyper {M}arkov laws in the statistical analysis of decomposable
  graphical models.
\newblock {\em The Annals of Statistics}, 21(3):1272--1317.

\bibitem[Del~Sol et~al., 2010]{del2010diseases}
Del~Sol, A., Balling, R., Hood, L., and Galas, D. (2010).
\newblock Diseases as network perturbations.
\newblock {\em Current Opinion in Biotechnology}, 21(4):566--571.

\bibitem[Dethlefsen and H{\o}jsgaard, 2005]{grbase}
Dethlefsen, C. and H{\o}jsgaard, S. (2005).
\newblock A common platform for graphical models in {R}: The {gRbase} package.
\newblock {\em Journal of Statistical Software}, 14(17):1--12.

\bibitem[Frydenberg and Lauritzen, 1989]{frydenberg1989decomposition}
Frydenberg, M. and Lauritzen, S.~L. (1989).
\newblock Decomposition of maximum likelihood in mixed graphical interaction
  models.
\newblock {\em Biometrika}, 76(3):539--555.

\bibitem[Goeman et~al., 2004]{goeman2004global}
Goeman, J.~J., Van De~Geer, S.~A., De~Kort, F., and Van~Houwelingen, H.~C.
  (2004).
\newblock A global test for groups of genes: testing association with a
  clinical outcome.
\newblock {\em Bioinformatics}, 20(1):93--99.

\bibitem[Griffin et~al., 2018]{griffin2018detection}
Griffin, P.~J., Zhang, Y., Johnson, W.~E., and Kolaczyk, E.~D. (2018).
\newblock Detection of multiple perturbations in multi-omics biological
  networks.
\newblock {\em Biometrics}, 74(4):1351--1361.

\bibitem[Hudson et~al., 2009]{hudson2009differential}
Hudson, N.~J., Reverter, A., and Dalrymple, B.~P. (2009).
\newblock A differential wiring analysis of expression data correctly
  identifies the gene containing the causal mutation.
\newblock {\em PLoS Comput Biol}, 5(5):e1000382.

\bibitem[Hummel et~al., 2008]{hummel2008globalancova}
Hummel, M., Meister, R., and Mansmann, U. (2008).
\newblock Global{ANCOVA}: exploration and assessment of gene group effects.
\newblock {\em Bioinformatics}, 24(1):78--85.

\bibitem[Kanehisa and Goto, 2000]{kanehisa:2000}
Kanehisa, M. and Goto, S. (2000).
\newblock {KEGG:} {K}yoto {E}ncyclopedia of {G}enes and {G}enomes.
\newblock {\em Nucleic Acids Research}, 28(1):27--30.

\bibitem[Lauritzen, 1996]{lauritzen:96}
Lauritzen, S.~L. (1996).
\newblock {\em Graphical models}.
\newblock Clarendon Press, Oxford.

\bibitem[Li, 2009]{ALL}
Li, X. (2009).
\newblock {\em ALL: A data package}.
\newblock R package version 1.16.0.

\bibitem[Maglott et~al., 2005]{maglott2005entrez}
Maglott, D., Ostell, J., Pruitt, K.~D., and Tatusova, T. (2005).
\newblock Entrez {G}ene: gene-centered information at {NCBI}.
\newblock {\em Nucleic Acids Research}, 33(suppl 1):D54--D58.

\bibitem[Ritchie et~al., 2015]{ritchie2015limma}
Ritchie, M.~E., Phipson, B., Wu, D., Hu, Y., Law, C.~W., Shi, W., and Smyth,
  G.~K. (2015).
\newblock limma powers differential expression analyses for rna-sequencing and
  microarray studies.
\newblock {\em Nucleic acids research}, 43(7):e47--e47.

\bibitem[Robins et~al., 2003]{robins2003uniform}
Robins, J.~M., Scheines, R., Spirtes, P., and Wasserman, L. (2003).
\newblock Uniform consistency in causal inference.
\newblock {\em Biometrika}, 90(3):491--515.

\bibitem[Sales et~al., 2016]{graphite}
Sales, G., Calura, E., and Romualdi, C. (2016).
\newblock {\em graphite: GRAPH Interaction from pathway Topological
  Environment}.
\newblock R package version 1.20.1.

\bibitem[Tan, 1977]{tan1977distribution}
Tan, W. (1977).
\newblock On the distribution of quadratic forms in normal random variables.
\newblock {\em Canadian Journal of Statistics}, 5(2):241--250.

\bibitem[Thomas and Green, 2009]{thomas2009enumerating}
Thomas, A. and Green, P.~J. (2009).
\newblock Enumerating the junction trees of a decomposable graph.
\newblock {\em Journal of Computational and Graphical Statistics},
  18(4):930--940.

\bibitem[Tsai and Chen, 2009]{tsai2009multivariate}
Tsai, C.-A. and Chen, J.~J. (2009).
\newblock Multivariate analysis of variance test for gene set analysis.
\newblock {\em Bioinformatics}, 25(7):897--903.

\bibitem[Westfall and Young, 1993]{westfall1993resampling}
Westfall, P.~H. and Young, S.~S. (1993).
\newblock {\em Resampling-based multiple testing: Examples and methods for
  p-value adjustment}, volume 279.
\newblock John Wiley \& Sons, New York.

\bibitem[Xia et~al., 2015]{xia2015testing}
Xia, Y., Cai, T., and Cai, T.~T. (2015).
\newblock Testing differential networks with applications to the detection of
  gene-gene interactions.
\newblock {\em Biometrika}, 102(2):247--266.

\bibitem[Zhao et~al., 2014]{zhao2014direct}
Zhao, S.~D., Cai, T.~T., and Li, H. (2014).
\newblock Direct estimation of differential networks.
\newblock {\em Biometrika}, 101(2):253--268.

\bibitem[Zhu and Bradic, 2016]{zhu2016two}
Zhu, Y. and Bradic, J. (2016).
\newblock Two-sample testing in non-sparse high-dimensional linear models.
\newblock {\em arXiv preprint arXiv:1610.04580}.

\end{thebibliography}

\appendix

\section{Undirected graphs basics}
\label{preliminaries}
Here, we briefly review key graph notions relevant for our work. For a detailed exposition, see \cite{lauritzen:96}.

Consider an undirected graph $G=(V,E)$ where  $V$ is a  set of nodes and  $E$ is a set of edges. A subset of vertices $A$ defines an induced subgraph $G_A=(A, E \cap A\times A)$.  A subgraph is said to be complete if all pairs of its vertices are connected in $G$. A clique is a maximal complete subgraph, that is, it is not a subgraph  of any other complete subgraph. 
Two disjoint subsets $A,B \subset V$ are said to be \textit{separated} by a subset $S$ (disjoint from $A$ and $B$) if all paths from $A$ to $B$ contain vertices from $S$. A graph $G$ is decomposable if and only if the set of cliques of $G$ can be ordered so as to satisfy the \textit{running intersection property}, that is, for every $i=2,\ldots,k$, if $S_i= C_i \cap \bigcup_{j=1}^{i-1} C_j $, then $ S_i \in C_l$, for some $l <i-1$.
Although this ordering is generally not unique,  the  structure of the graph uniquely determines the set of cliques $\left\{C_1,\ldots,C_k\right\}$ and the set of \textit{separators} $\left\{S_2,\ldots, S_k\right\}$. For ease of notation, it is often set $S_1=\varnothing$, so that the set of separators becomes $\left\{S_1,\ldots, S_k\right\}$.

\section{Graphical seed set: illustrative example}\label{example}

We use a small undirected graph $G$  shown in Figure \ref{Smallex} to illustrate  
possible  relations between the minimal seed set and the graphical seed set.  Graph $G$ consists of  cliques $C_1=\left\{1,2,3\right\}$ and $C_2=\left\{3,4,5\right\}$ separated by $S=\left\{3\right\}$. In the left panel, the minimal seed set $D=\left\{3\right\}$ coincides with the separator $S$, and thus with the graphical seed set as well.  In the middle panel, the minimal seed set is $D=\left\{1,3\right\}$. Node $2$ is not separated from $D$ by any separator in $G$ (in this case, neither $S$ nor empty set). Nodes 4 and 5 are separated from $D$ by $S$, since all paths from 4 and 5 to $D$ pass through $S$. The graphical seed set is thus $D_G=\left\{1,2,3\right\}$. In the right panel, the minimal seed set is $D=\left\{1,4\right\}$. None of the remaining nodes 2, 3 and 5 is separated from $D$ by a separator in $G$, and so the graphical seed set is the entire set of nodes $D_G= \left\{1,2,3,4,5\right\}$. \qed
	\begin{figure}
		\begin{minipage}{0.3\textwidth}
			\center
			\begin{tikzpicture}[xscale=0.35,yscale=0.35,
			inner sep=1.1mm, auto]
			\draw (0,0) node (3) [draw, circle, fill=Crimson]  {\small$3$}
			node (1) [draw] at (-2.5,3.5) [circle, draw, fill =Silver] {\small$1$}
			node(2) [draw] at (2.5,3.5) [circle, draw,  fill =Silver] {\small$2$}
			node(4)[draw] at (-2.5,-3.5) [circle,fill =Silver] {\small$4$}
			node(5)[draw] at (2.5,- 3.5) [circle,fill =Silver] {\small$5$};
			\draw   (1) -- (2); 
			\draw (1) -- (3);
			\draw    (2) -- (3);
			\draw    (3) -- (4);
			\draw   (3) -- (5);
			\draw   (4) -- (5);
			\end{tikzpicture}
		\end{minipage}	\begin{minipage}{0.3\textwidth}
		\center
		\begin{tikzpicture}[xscale=0.35,yscale=0.35,
		inner sep=1.1mm, auto]
		\draw (0,0) node (3) [draw, circle, fill=Crimson]  {\small 3}
		node (1) [draw, circle, fill=Crimson] at (-2.5,3.5) [circle] {\small $1$}
		node(2) [draw] at (2.5,3.5) [circle, fill= Salmon] {\small$2$}
		node(4)[draw] at (-2.5,-3.5) [circle,fill =Silver] {\small$4$}
		node(5)[draw] at (2.5,- 3.5) [circle,fill =Silver] {\small$5$};
		\draw    (1) -- (2); 
		\draw   (1) -- (3);
		\draw   (2) -- (3);
		\draw (3) -- (4);
		\draw  (3) -- (5);
		\draw    (4) -- (5);
		\end{tikzpicture}
	\end{minipage} \begin{minipage}{0.3\textwidth}
	\center
		\begin{tikzpicture}[xscale=0.35,yscale=0.35,
	inner sep=1.1mm, auto]
	\draw (0,0) node (3) [draw, circle, fill=Salmon]  {\small 3}
	node (1) [draw, circle, fill=Crimson] at (-2.5,3.5) [circle] {\small $1$}
	node(2) [draw] at (2.5,3.5) [circle, fill= Salmon] {\small$2$}
	node(4)[draw] at (-2.5,-3.5) [circle,fill =Crimson] {\small$4$}
	node(5)[draw] at (2.5,- 3.5) [circle,fill =Salmon] {\small$5$};
	\draw   (1) -- (2); 
	\draw    (1) -- (3);
	\draw    (2) -- (3);
	\draw     (3) -- (4);
	\draw   (3) -- (5);
	\draw   (4) -- (5);
	\end{tikzpicture}
\end{minipage}
\caption{Minimal seed sets (dark red) and associated graphical seed sets (difference between the two in light red).} \label{Smallex}
\end{figure}

The above example illustrates that  $D_G$ might be larger than the set of interest, i.e.  the minimal  seed set $D$.   In most situations, however, the graphical seed set will allow us to zoom in on the set $D$,  while exploiting the modularity of the graphical structure.

\section{Technical details and proofs}

\subsection*{Proof of Proposition \ref{gss}}
Let $P= \bigcap_{_{i=1}}^{k}\bigcup_{\left\{j: \,\,d_{i,j}^*=1\right\}} C_{i,j}$. %
Then if $ v\in P$, for each decomposition $i$, there is at least one clique $C_{i,j}$ containing $v$ such that $d_{ij}^*=1$. If $C_{i,l}$ denotes the first clique in the $i$-th decomposition containing $v$,  we know that $v$ belongs to $R_{i,l}$, otherwise $C_{i,l}$ would not be the first clique containing $v$. Consider a tree of cliques constructed from the perfect ordering $C_{i,1}, \ldots, C_{i,k}$ in the following fashion. The perfect ordering property guarantees that for each clique $C_{i,j}$, the intersection with the union of predecessor cliques  is contained within a single clique, that is
\begin{equation}\label{perfect}
	C_{i,j} \cap \bigcup_{m=1}^{j-1} C_{i,m} \subset C_{i,n}, \quad \mbox{for some } n=1,\ldots,j-1.
\end{equation} 
Then set $C_{i,n}$ to be a parent of $C_{i,j}$ in the clique tree.  Parent clique might not be unique, but without loss of generality, we take   the first clique  satisfying the assumption \eqref{perfect}. Then all cliques containing $v$ other than $C_{i,l}$ must be descendants of $C_{i,l}$. 
We further notice that  if $d_{i,l}^*=0$, then $d_{i,m}^*=0$ for all its descendants. 
This implies that necessarily $d_{i,l}^*=1$ and $S_{i,l}$ does not separate $v$ from $D$. Since this is true for all decompositions, there can be no separator that separates $v$ from $D$, implying that $v$ belongs to  $
 D_G$.  
 
We have proven $v\in P \Rightarrow v\in D_G$, but all considered implications remain valid if reversed, so that $v\in P \Leftrightarrow v\in D_G$.
\qed
\\

\subsection*{Proof of Proposition \ref{prop2}}
Choose $\alpha_n=(1-F_U(n^d))$, with $0<d<1/2$, $U \sim \chi^2_f$, and let $q_n = F_U^{-1}(\alpha_n).$ 
Under the null hypothesis, $\lambda_{i,j;n} \stackrel{d}{\rightarrow} \lambda$, with $\lambda \sim \chi^2_f.$ Thanks to the Slutsky theorem, we can write
$$\mathbb{P}_{(\theta^{(1)}, \theta^{(2)}) \in \Theta_0}(\phi_{i,j}(n) = 1)= \mathbb{P}_{(\theta^{(1)}, \theta^{(2)})\in \Theta_0}\left(\frac{\lambda_{i,j;n}}{n^d} > 1\right)  \longrightarrow 0. $$

Furthermore,  for each $(\theta^{(1)}, \theta^{(2)}) \in \Theta_1,$ it is known that the log likelihood ratio test
is degenerate with the order $O(\sqrt{n}).$ With the choice of $\alpha_n$ above,
$$\mathbb{P}_{(\theta^{(1)}, \theta^{(2)})\in \Theta_1}(\phi_{i,j}(n) = 0) 	=\mathbb{P}_{(\theta^{(1)}, \theta^{(2)}) \in \Theta_1}\left(\frac{\lambda_{i,j;n}}{n^d} < 1\right)  \longrightarrow 0. \qed $$

 \subsection*{Proof of Theorem \ref{t2}}
	For a fixed $i,$ we have that $\phi_i(n) = (\phi_{i,1}(n), \ldots, \phi_{i,k}(n)) \rightarrow d_i^* = (d_{i,1}^*, \ldots,  d_{i,k}^*),$  since  the  inequality $$\mathbb{P}_{(\theta^{(1)}, \theta^{(2)})}(\phi_{i}(n)=d^*_{i})\geq 1- \sum_{j=1}^k\mathbb{P}_{(\theta^{(1)}, \theta^{(2)})}(\phi_{i,j}(n) \neq d^*_{i,j}) $$ in conjunction with Proposition   \ref{prop2} implies   $\mathbb{P}_{(\theta^{(1)}, \theta^{(2)})}(\phi_{i}(n)=d^*_{i})\longrightarrow 1$.  Convergence of $\hat{D}_G$ to $D_G$ follows straightforwardly.\qed
	
{\color{black}\section{Simulation studies}

\subsection{Skew-normal graphical models}\label{SMSkew}
To investigate the question of robustness of the proposed method in the Gaussian context, we conducted a simulation study with data sampled from a skew-normal graphical model \citep{capitanio2003graphical}. We recall that a $p$-dimensional random vector $X$ is said to follow a multivariate skew-normal distribution if its density is of the form \citep{azzalini1999statistical}:
$$
\phi_p(x;\mu, \Omega) \Phi\left(\alpha_0 + \alpha^T \omega^{-1}(x-\mu)\right)/ \Phi(\tau), \quad x \in \mathbb{R}^p,
$$
where \begin{itemize}
\item $\phi_p(x;\mu, \Omega)$ is the probability density function of the $p$-dimensional normal distribution $N_p(\mu, \Omega)$;
\item $\Phi$ is the cumulative distribution function of the standard normal distribution $N(0,1)$;
\item $\mu \in \mathbb{R}^p$, $\tau \in \mathbb{R}$ and $\Omega$ is a $p\times p$ full rank variance matrix;
\item $\omega = \text{diag}\left(\Omega_{11}, \Omega_{22}, \ldots, \Omega_{pp} \right)^{1/2}$
\item $\alpha\in \mathbb{R}^p$ is a shape parameter and $\alpha_0 = \tau(1+\alpha^T\omega^{-1}\Omega\omega^{-1}\alpha)^{1/2}$.
\end{itemize}
\cite{capitanio2003graphical} showed that $X_i$ and $X_j$ are conditionally independent given the remaining components of $X$ if  and only if
\begin{equation}
\label{sngm}
\Omega^{ij} =0 \quad \text{and} \quad \alpha_i\alpha_j=0,
\end{equation}
where $\Omega^{ij}$ is the element $(i,j)$ of the matrix $\Omega^{-1}$. }

{\color{black}We considered  graph $G$ of Simulation study 2, shown also in Figure \ref{graph_ss3}.}
\begin{figure}\centering
\includegraphics[width = 0.6\linewidth]{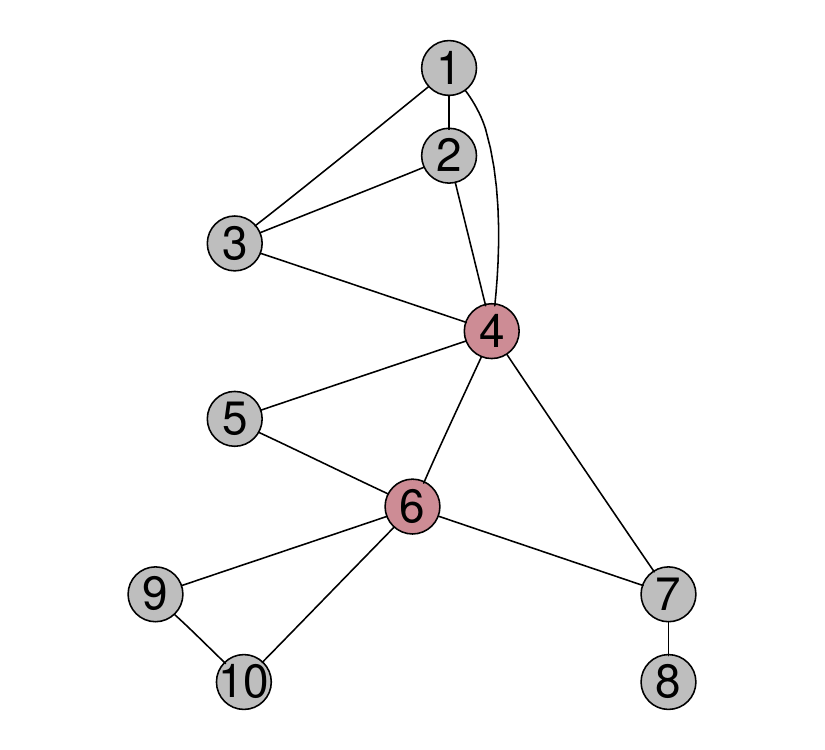}
\caption{\label{graph_ss3}Graph used in the skew-normal simulation study. Seed set is highlighted. }
\end{figure}
{\color{black}The seed set was set to $D=\left\{4,6\right\}$. Components of the location parameter $\mu$ were drawn from $N(0.5, 1)$. Matrix $\Omega$ was obtained from a matrix with 1s on a diagonal and $0.6$ off the diagonal that was modified so that the its inverse reflects the missing edges of $G$. In the second condition, the location parameter of the seed set variables $(\mu_4,\mu_6)$ was multiplied by a $1.5$ and their scale parameter $(\Omega_{44}, \Omega_{66})$  was decreased by $50$\%. The parameter of skewness $\alpha\in \mathbb{R}$, assumed  shared across the two conditions,  varied in the  set $\left\{0, 1, 2, 4, 8, 16, 20 \right\}$.  In particular, the skewness of variables $X_1, \ldots,X_6, X_9,X_{10}$ was set to  $\alpha$ or $-\alpha$ with the sign randomly chosen, while marginal distributions of $X_7$ and $X_8$ were symmetric, so that the condition \eqref{sngm} is satisfied for all pairs of nodes not connected in $G$,  ensuring that the conditional independence relations reflected in $G$ remain preserved. Note that the case $\alpha = 0$ corresponds to the normal distribution and allows us to study the impact of skewness.  The marginal distributions of the ten variables for $\alpha =8$ is shown in Figure \ref{ss3_2}.
\begin{figure}\centering
\includegraphics[width = 0.9\linewidth]{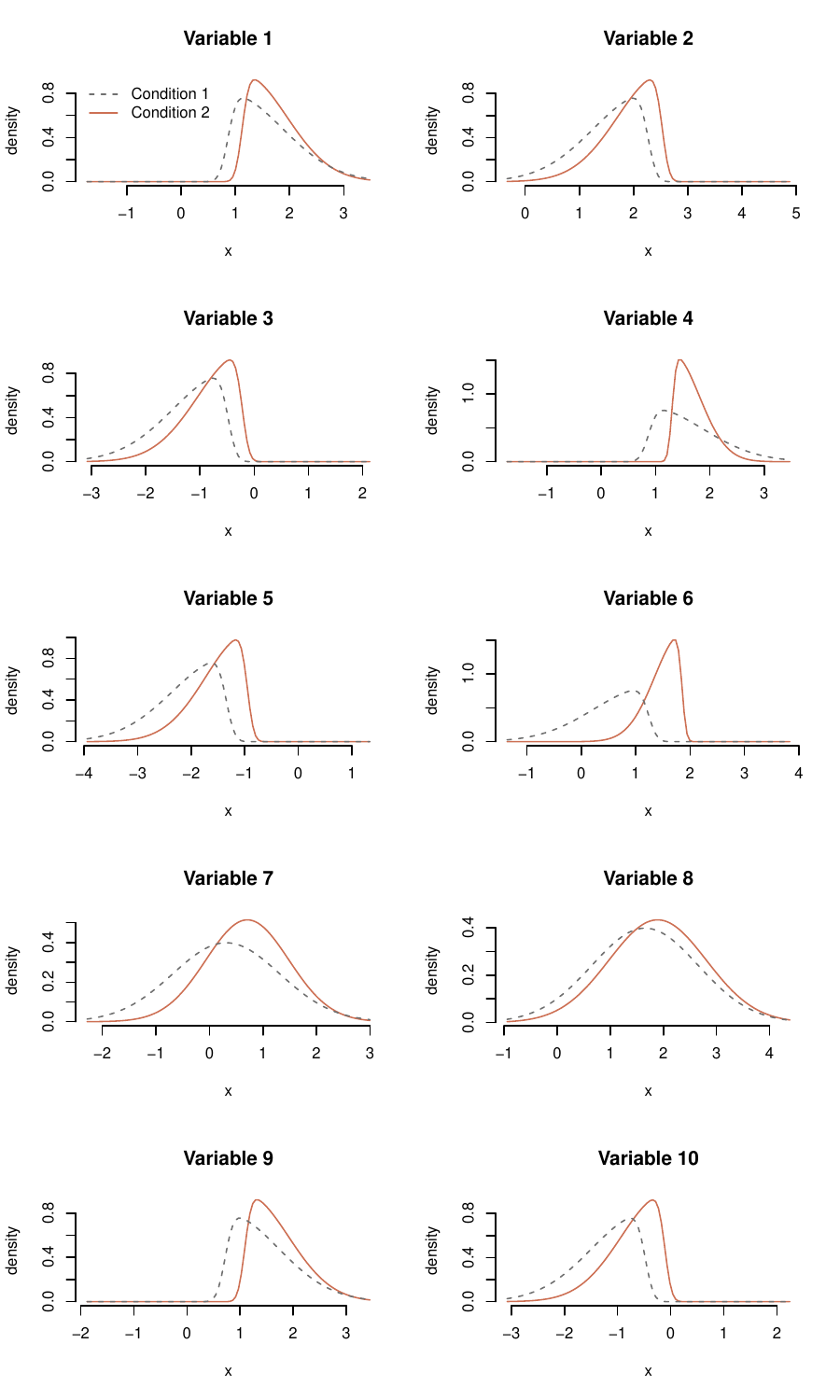}
\caption{\label{ss3_2}Marginal distributions of the 10 variables for $\alpha = 8$.}
\end{figure}

We generated random samples from multivariate skew-normal distributions with  R package \texttt{sn} \citep{skewnormalR}.
We considered three sample sizes $n_1=n_2 \in \left\{50, 100, 200\right\}$, and for each sample size we generated 500 pairs of datasets. As before, to evaluate the performance of the seed set estimating procedure, we computed the empirical power, defined as the frequency with which the seed set was correctly identified, i.e. $\hat{D}_G = D_G$, and the empirical familywise error rate, defined as the frequency with which $\hat{D}_G$ contained a false positive. Figure \ref{ss3_2} displays the results.

\begin{figure}\centering
\includegraphics[width = 0.75\linewidth]{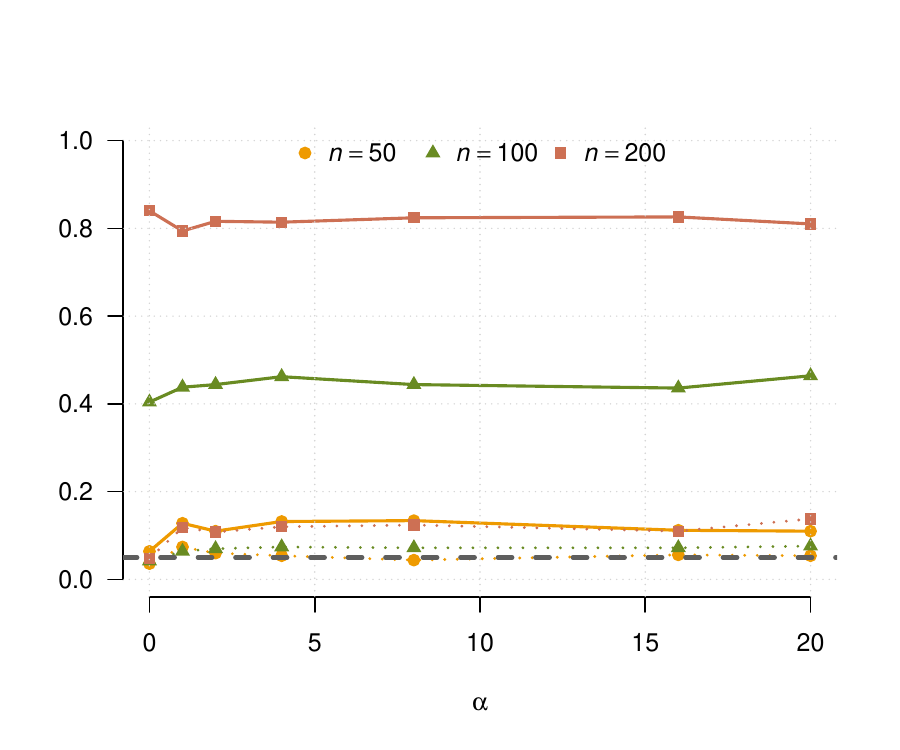}
\caption{\label{ss3}Empirical power (solid lines) and familywise error rate (dotted lines) of the graphical seed set estimating procedure as a function of the skewness parameter $\alpha$. Dashed horizontal line $y=0.05$, representing the nominal familywise error rate, is added for reference.}
\end{figure}

As expected, the empirical power is increasing with increasing sample size. More interestingly, the power does not seem to be much affected  by the skewness. On the other hand, the familywise error rate control is  compromised, but the increase is so slight  that it allows us to infer that the seed set estimating procedure is quite robust in the presence of skewness. 

It should be stressed that extra skewness is only one of the many forms that departures from normality can take. Nevertheless, when studying the properties of procedures in the graphical modelling context,  the family of skew-normal distributions has an important advantage over other continuous multivariate distributions:  we can explicitly, through  restrictions on the parameter space, link conditional independence relations with an undirected graph.  When this is not the case, it is difficult to disentangle the effect of non-normality from other forms of misspecification.  }

{\color{black}	
\subsection{Comparison with the network filtering approach of \cite{griffin2018detection}\label{SMcomparison}}

As already mentioned in the article, to the best of our knowledge, there are no methods that aim to estimate the seed set, as defined in this work. There are, however, approaches that aim to detect the origin of more specific types of perturbations. For instance,  \cite{griffin2018detection} focus on perturbations that affect the mean level. The Authors propose to search for the target of perturbation by applying the method of network filtering. They further propose a sequential multiple testing procedure for identifying  multiple perturbation targets.  The approach is implemented in the R package \texttt{mapggm} available from \texttt{https://github.com/paulajgriffin/mapggm}. In  what follows, we briefly describe the approach  and the assumed  perturbation model. 

Data in the control condition are assumed to come from a multivariate normal distribution that is Markov with respect to an unknown graph. The perturbation acts on its target(s) and changes its(their) mean. The effect of perturbation is then propagated through network connections so that further nodes result perturbed. The aim of detecting the site of the original perturbation is achieved in two steps. In the first step, data from the first condition are  used to estimate the covariance matrix and the graphical structure; in the second step, data from the second, i.e. perturbed, condition are transformed in the process of network filtering, and a testing procedure is used to identify the most likely sites of the original perturbation. 

To compare the seed set approach with the approach based on network filtering, we performed a simulation study based on the  graph $G$ shown in Figure \ref{graph_ss3}. We again set the seed set to $D=\left\{4,6\right\}$, but  in this case we perturbed the means of the two variables. In particular, data from the first condition are simulated from ${\sf N}(0,\Sigma)$, where 
$\Sigma$ is the covariance matrix obtained from a matrix with 1s on the main diagonal and $0.6$ off diagonal, modified so that its inverse has zeroes corresponding to the missing edges of $G$. Data from the second condition come from ${ \sf N}(\Sigma\mu, \Sigma)$, where $\mu\in \mathbb{R}^{10}$, such that its elements are equal to $\delta \in \mathbb{R}$ if they correspond to the perturbation targets, i.e. seed set,  and $0$ otherwise. 
Parameter  $\delta$ varied in the set  $\left\{0.5, 1, 2, 4, 8, 16\right\}$. 

When applying the network filtering approach, instead of estimating
network structure encoded in $\Sigma$ via penalized regression, we used the information
on the structure of  $G$, so that the comparison with the seed set approach is more balanced. For each $\delta$, we generated 1000 pairs of datasets with $n_1=n_2 = 50$.  We controlled familywise error rate at  $\alpha=0.05$; for the seed set approach with the max$T$ method as described in Section 3.4 of the article,   for the network filtering approach with the Bonferroni 
correction applied to the  node-wise $p$-values. 

We computed the empirical power for the two methods defined as the frequency with which 
\begin{itemize}
\item the true seed set was either correctly identified or covered by the seed set estimate;
\item  the set of detected perturbation targets, defined as a set of nodes with $p_{\text{adj}}\leq\alpha=0.05$, covered the true seed set. 
\end{itemize}
Similarly, the familywise error rate was estimated as the frequency with which the seed set estimate contained a false positive, and the frequency with which the set of detected perturbation targets included a false positive. The results are shown in Table \ref{multi}. 
 }

 \begin{table}[h]
	\caption{\label{mtp}Empirical power and familywise error rate  multiplied by $10^2$ for the seed set and the network filtering approach \citep{griffin2018detection}. Estimated familywise error rate exceeding the nominal level is highlighted.}
	\label{multi}
	\setlength{\tabcolsep}{10pt}
	\begin{center}
		\begin{tabular}{r|@{\hspace{0.25in}}rr@{\hspace{0.25in}}rr}
			& \multicolumn{2}{c}{Seed set} &  \multicolumn{2}{c}{Network filtering}\\
			\hline
			$	\delta$	 & Power & FWER & Power  & FWER \\
			\hline
			0.5	 &	      1.5   &    $3.5$    &    16.8   &  $\mathbf{39.5}$\\
			1   &         8.2   &    $3.9$   &   72.9   & $\mathbf{75.3}$\\
			2    &	    40.0   &    ${5.0}$   &   94.6  &  $\mathbf{98.8}$ \\
			4   &	    61.9   &    ${4.1}$     &   98.5 & $\mathbf{1.0}$\\
			8    &         70.0   &    ${3.8}$     &   98.8  &  $\mathbf{1.0}$\\
			16   &         71.8   &    ${3.6}$     &   99.1  &  $\mathbf{1.0}$\\
               \end{tabular}
	\end{center}
\end{table}
{\color{black}
The network filtering approach has more power than the seed set approach, with a particularly  striking difference  for the low values of $\delta$. However, the power advantage comes at the cost of losing type 1 error control: the actual familywise error rate for the network filtering approach is always above the nominal level $\alpha=0.05$. Furthermore, it  quickly reaches 1, which implies that for $\delta$ large enough, the set of detected targets will almost surely contain at least one false positive. A closer inspection shows that this behaviour is at least partially due to the estimation of $\Sigma$. Namely, the estimate obtained from the first condition is used in the second step of network filtering as a plug in estimate. As a consequence, although this strategy has asymptotic guarantees, in finite samples it  can lead to a significant inflation of  the type I error rate, as evidenced by this example. }

\begin{figure}[htb]
		\centering
		\makebox{\includegraphics[width=0.99\textwidth]{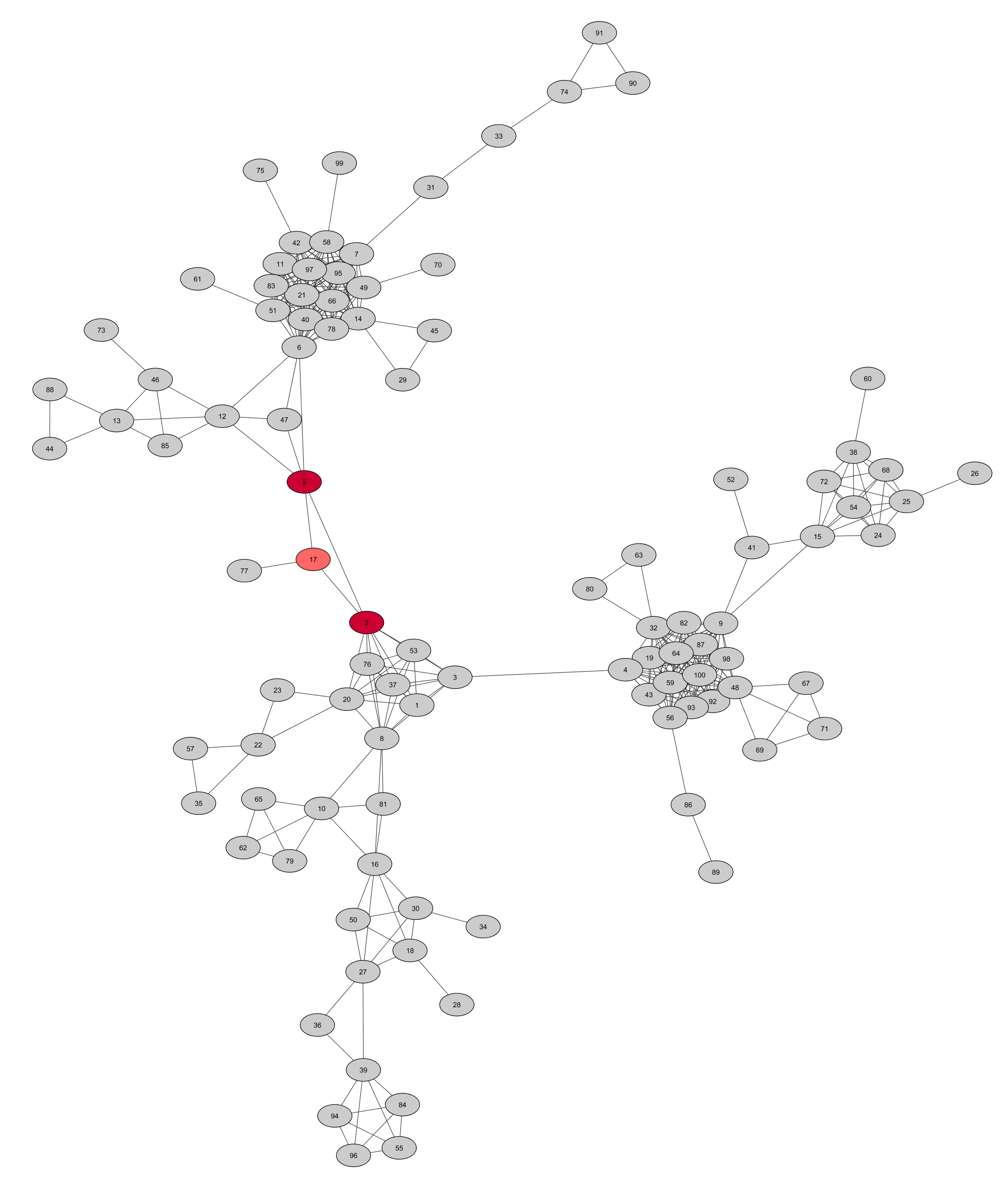}}
		\caption{Undirected graph used in  Simulation study 1.  The minimal seed set is set to $D=\left\{2,5\right\}$, shown in dark red, with the corresponding  graphical seed set $D_G=\left\{2,5,17\right\}$.}
		\label{fig:SimulationStudy}
\end{figure}

\begin{figure}
		\centering
		\begin{minipage}[t]{0.45\textwidth}
			\includegraphics[width=0.890\textwidth]{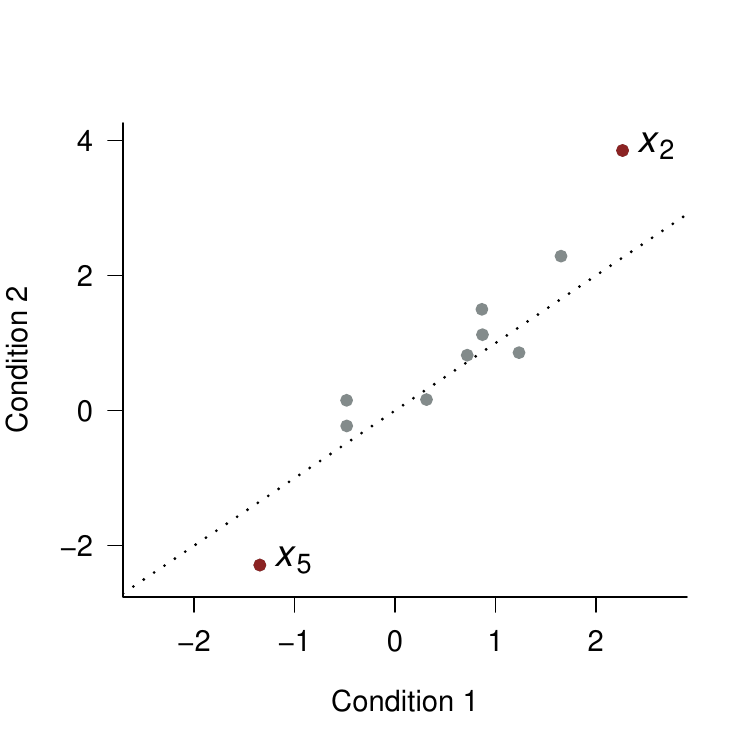}\\
		\end{minipage}
		\begin{minipage}[t]{0.45\textwidth}
			\includegraphics[width=0.890\textwidth]{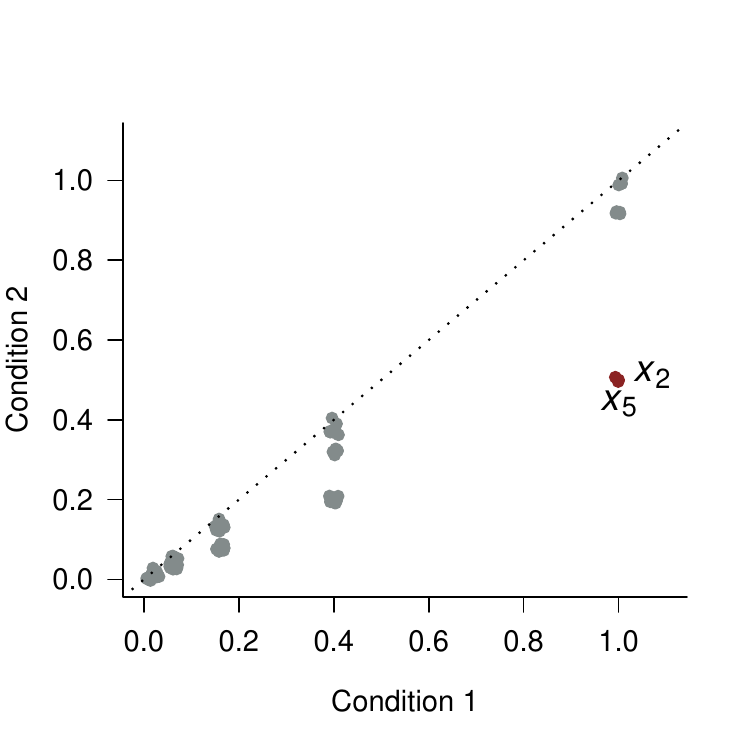}
		\end{minipage}
		\caption{Simulation study 1:  comparison of the parameters in two conditions. On the left, the means  of the first 10 variables, on the right, the associated variances.  Means and variances of the seed set variables are highlighted in red. A dotted $y=x$ line is added for reference. A small noise is added to the plotted points on the right to avoid  a complete overlap.}
		\label{vcmat}
\end{figure}

\begin{figure}
		\centering
		\includegraphics[width=0.9\textwidth]{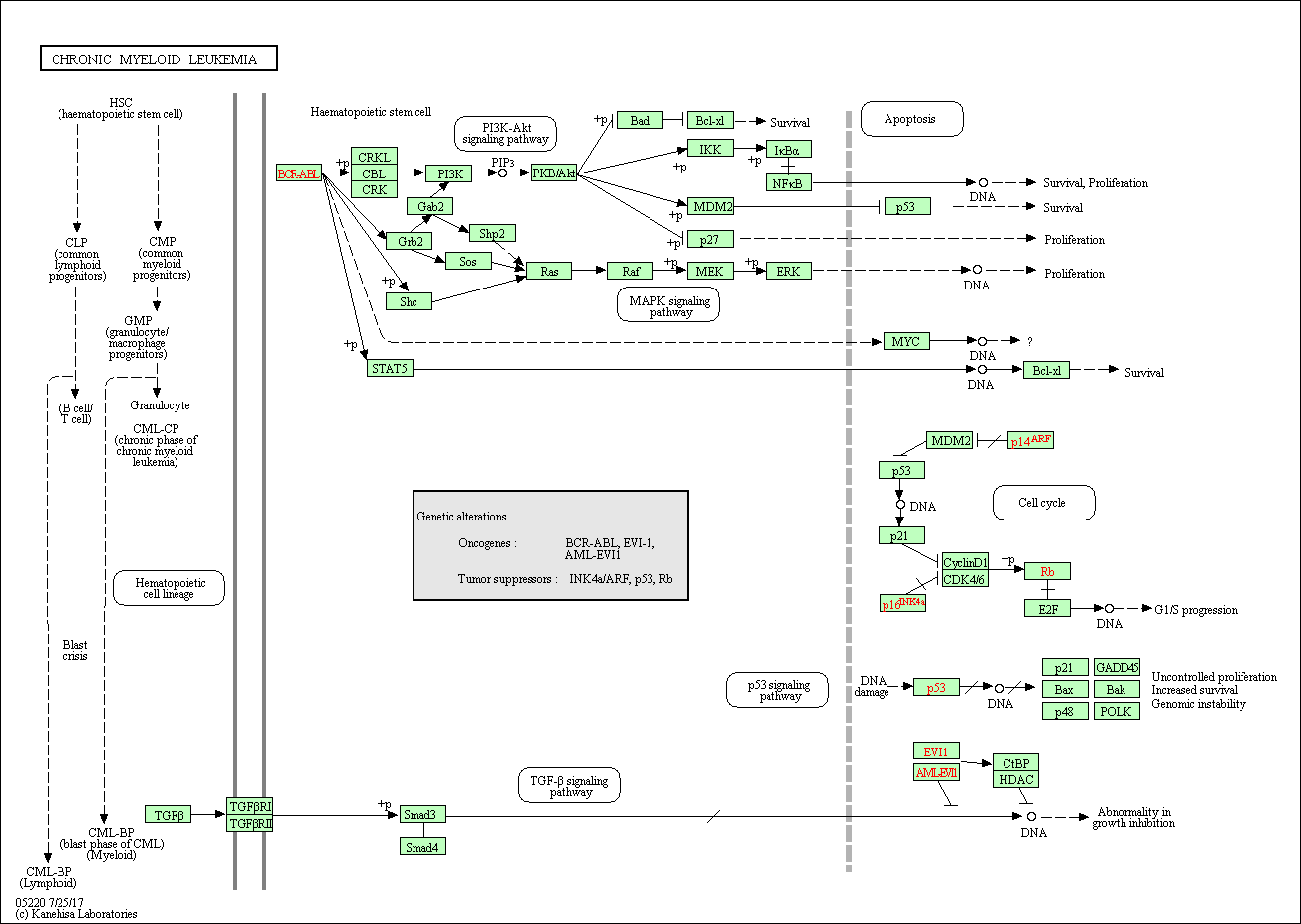}
		\caption{Chronic myeloid leukemia pathway from KEGG.}
		\label{fig:ALLpathway}
\end{figure}

\begin{table}
{\small
	\centering
		\caption{\label{tab_chimera2} Chronic myeloid leukemia dataset: results of local tests. 
			}
			\setlength{\extrarowheight}{.3em}
				\begin{tabular}{>{$}l<{$} >{$}l<{$} >{$}r<{$} >{$}l<{$} >{$}l<{$} >{$}r<{$}}
				\noalign{\vskip 0.1cm}
				\hline
				\noalign{\vskip 0.1cm}
				{\bf No.}&{\bf Test}& \boldsymbol p{\bf -value} &  {\bf No.} &{\bf Test}& \boldsymbol p\,{\bf -value}\\
				\noalign{\vskip 0.2cm}
				\hline
				1 &  1398,1399,25,613,867,9846 & 6.0\times10^{-4} & 21 & 25,613,6777 & 6.0\times 10^{-4}    \\ 
				2 & 5295,8503|1398,1399,867,9846 & 3.9\times 10^{-1}      & 22 & 25,25759,613 & 6.0\times 10^{-4}  \\ 
				3 & 2885|25,613,9846& 9.5\times 10^{-1}                     & 23 & 25,4609,613 &6.0\times 10^{-4}  \\ 
				4 &  207|5295,8503   & 9.2\times 10^{-2}                    & 24 & 1147,207,3551 & 5.7\times10^{-1}\\ 
				5 & 6776|25,613   & 2.4 \times 10^{-3}                       & 25 & 5295,8503|207 &8.2\times10^{-2} \\ 
				6 &  6777|25,613 & 9.3\times 10^{-1}
				& 26 & 1398,1399,867,9846|5295,8503 &9.3\times 10^{-1}\\ 
				7 & 25759|25,613 & 8.4 \times 10^{-1}                      & 27 & 25,613|1398,1399,867,9846 &  6.0\times 10^{-4} \\ 
				8 &  4609|25,613 & 1.7\times 10^{-1}
				& 28 & 207,4193 &4.4\times 10 ^{-1}\\ 
				9 & 1147,3551|207 & 6.2\times 10^{-1}                     & 29 & 207,5295,8503  &8.4\times 10^{-2}  \\ 
				10& 4790,4792|1147,3551  & 1.3 \times 10^{-2}               & 30 & 1147,3551,4790,4792 & 5.0\times 10^{-2}\\ 
				11 & 6654,6655|2885  &3.6\times 10 ^{-1}                    & 31 &207|1147,3551 &4.4\times 10^{-1}\\ 
				12 &3265,3845,4893|6654,6655  & 9.8\times 10^{-1}          & 32 &3265,3845,4893,6654,6655 &8.8\times 10^{-1}\\ 
				13 & 369|3265,3845,4893  & 5.6\times 10^{-1}              & 33 &2885|6654,6655   &  9.6\times10^{-1}\\ 
				14 & 5894|3265,3845,4893   & 5.1\times10^{-1}             & 34 &       25,613,9846|2885  & 6.0\times 10^{-4}\\ 
				15 & 4193|207 & 3.3\times10^{-2}                          & 35 & 3265,3845,4893,5894 & 6.5\times 10^{-3} \\ 
				16 & 7157|4193 & 1.4\times 10^{-1}                         & 36 & 6654,6655|3265,3845,4893 & 9.2 \times 10^{-1} \\ 
				17 & 25,2885,613,9846 & 6.0\times 10^{-4}                  & 37 & 3265,369,3845,4893 & 6.8\times 10^{-1} \\ 
				18 & 1398,1399,867|25,613,9846 & 4.8\times 10^{-1}         & 38 &4193,7157&  1.3\times 10^{-2}\\ 
				19 & 25,613,6776 & 6.0\times 10^{-4}                        & 39 & 207|4193 & 4.4\times10^{-1}\\ 
				20 & 1398,1399,867,9846|25,613 &  3.6\times 10^{-1}         & 40 & 1398,1399,5295,8503,867,9846& 8.0\times 10^{-1}\\ 
				&  &                                           & 41 & 2885,6654,6655 & 5.4\times 10^{-1}\\ 
				\hline
			\end{tabular}}
		\end{table}

\end{document}